\documentclass[12pt]{article}
\usepackage[utf8]{inputenc}

\usepackage[textwidth=6.8in,top=1.18in,bottom=1.18in]{geometry}
\usepackage[margin=1cm]{caption}
\title{\vspace{-3cm} Semiparametric proximal causal inference}

\author{Yifan Cui\thanks{Center for Data Science, Zhejiang University}, Hongming Pu \thanks{Department of Statistics and Data Science, The Wharton School, University of Pennsylvania}, Xu Shi\thanks{Department of Biostatistics, University of Michigan}, Wang Miao\thanks{Department of Probability and Statistics, Peking University}, Eric Tchetgen Tchetgen\footnotemark[2]}

\date{}

\usepackage{amsthm,amsmath,amssymb,bbm,bm}
\usepackage{natbib}
\usepackage{multirow}
\usepackage[pdftex]{graphicx}
\usepackage{wrapfig}
\usepackage{tikz}
\usepackage{centernot}
\usepackage{array}
\usepackage{url}
\usepackage{algorithmic}
\usepackage[linesnumbered, ruled, vlined]{algorithm2e}
\usepackage{mathrsfs}
\usepackage{dsfont}
\usepackage{titling}
\usepackage{relsize}
\usepackage{rotating}
\usepackage{enumitem}
\usepackage{titling}
\usepackage{float}
\usepackage{subcaption}

\newtheorem{remark}{Remark}
\newtheorem{assumption}{Assumption}

\newtheorem{theorem}{Theorem}[section]

{\textwidth=6.5in}

\newcommand{\cX}{{\mathcal{X}}}

\newcommand{\logit}{{\text{logit}}}

\newcommand{\cM}{{\mathcal{M}}}

\newcommand{\PP}{{\mathbbm{P}}}
\newcommand{\cO}{{\mathcal{O}}}

\newcommand{\E}{\mathbb{E}}
\newcommand{\mM}{\mathcal{M}}

\newcommand{\yifan}[1]{\textcolor{black}{#1}\xspace}
\newcommand{\yf}[1]{\textcolor{black}{#1}\xspace}

\newcommand{\source}[1]{#1}

\begin{document}
\maketitle
\vspace{-1cm}
\abstract{
Skepticism about the assumption of no unmeasured confounding, also known as exchangeability, is often warranted in making causal inferences from observational data; because exchangeability hinges on an investigator's ability to accurately measure covariates that capture all potential sources of confounding.
In practice, the most one can hope for is that covariate measurements are
at best proxies of the true underlying confounding mechanism operating in a given observational
study. In this paper, we consider the framework of proximal causal inference introduced by \cite{miao2018identifying,tchetgen2020}, which while explicitly acknowledging covariate measurements as imperfect proxies of confounding mechanisms, offers an opportunity to learn about causal effects in settings where exchangeability on the basis of measured covariates fails. We make a number of contributions to proximal inference including (i) an alternative set of conditions for nonparametric proximal identification of the average treatment effect; (ii) general semiparametric theory for proximal estimation of the average treatment effect including efficiency bounds for key semiparametric models of interest; (iii) a characterization of proximal doubly robust and locally efficient estimators of the average treatment effect. Moreover, we provide analogous identification and efficiency results for the average treatment effect on the treated. 
Our approach is illustrated via simulation studies and a data application on evaluating the effectiveness of right heart catheterization in the intensive care unit of critically ill patients.}

\vspace{0.1cm}

\noindent {\bf keywords:}
Double robustness, Efficient influence function, Identification, Proximal causal inference, Semiparametric theory, Unmeasured confounding  

\section{Introduction}\label{sec:intro}

A common assumption for causal inference from observational data is that of no unmeasured confounding, which states that one has measured a
sufficiently rich set of covariates to ensure that within covariate strata, subjects are exchangeable
across observed treatment values.  Skepticism about such exchangeability assumption in observational studies is often warranted because it essentially requires investigators to accurately measure
covariates capturing all potential sources of confounding. 
In practice, confounding mechanisms
can rarely be learned with certainty from measured covariates. One may therefore only hope that covariate measurements are at best proxies of true underlying confounders.

There is a growing literature on formal causal inference methods that leverage certain types of proxies known as negative control variables to mitigate confounding bias in analysis of observational data \citep{lipsitch2010negative,kuroki2014measurement,miao2018identifying,Shi2019MultiplyRC,shi2020selective}. Existing negative control methods rely for point identification of causal effects on fairly restrictive assumptions such as linear models for the outcome and the unmeasured confounder \citep{flanders2011method,GagnonBartsch2012UsingCG,Flanders2015ANM,wang2017}, rank preservation \citep{tchetgen2013}, monotonicity \citep{sofer2016negative}, or categorical unmeasured confounders \citep{Shi2019MultiplyRC}. \cite{miao2018identifying} stands out in this literature as they formally establish sufficient conditions for nonparametric identification of causal effects using a pair of treatment and outcome proxies in the point treatment setting.

Building on  \cite{miao2018identifying}, \cite{tchetgen2020} recently introduced a potential outcome framework for proximal causal inference, which offers an opportunity to learn about causal effects
in point treatment and time-varying treatment settings where exchangeability on the basis of measured covariates fails. 
Proximal causal inference essentially
requires that the analyst can correctly classify a subset of measured covariates into three bucket types: 1) variables that may be common causes of the treatment and outcome variables; 2) treatment-inducing confounding
proxies versus; 3) outcome-inducing confounding proxies. A proxy of type 2) is a potential cause of
the treatment which is related with the outcome only through an unmeasured common cause for which the variable is a proxy; while a proxy of type 3) is a potential cause of the outcome which
is related with the treatment only through an unmeasured common cause for which the variable
is a proxy. Proxies that are associated with an unmeasured confounder but that are neither causes of treatment or outcome variables can belong to either
bucket type 2) or 3). 

Examples of proxies of type 2) and 3) abound in
observational studies. For instance, in an observational study evaluating the
effects of a treatment on disease progression, one is typically concerned that
patients either self-select or are selected by their physician to take the
treatment based on prognostic factors for the outcome; therefore there may be
two distinct processes contributing to a subject's propensity to be treated.
In an effort to account for these sources of confounding, a diligent
investigator would endeavor to record biomarker lab measurements and other clinically
relevant covariate data. Lab measurements of biomarkers are well-known to be error prone and therefore to at best serve as proxies of patients' underlying biological mechanisms at the source of confounding. Even when such measurements are available to the physician for treatment decision making, they seldom constitute a cause of disease progression (e.g., CD4 count or viral load in HIV care), but may be strongly associated with the latter to the extent that they share an unmeasured common cause (e.g., CD4 count is a proxy of underlying immune system status), and therefore may be viewed as proxies of type 2). As discussed in \cite{shi2020selective}, pre-treatment variables that satisfy the three core instrumental variable (IV) conditions (IV relevance, IV independence and exclusion restriction) constitute valid proxies of type 2); and in fact remain valid proxies even if IV independence assumption is violated  \citep{shi2020selective}.  A prominent proxy of type 3) often entails a baseline measurement of the outcome process, the basis of which serves as justification of the widely used difference-in-differences approach to account for confounding bias under no interaction assumptions or monotonicity conditions \citep{sofer2016negative}. Lifestyle choices such as exercising, alcohol use, smoking behavior, nutritional habits, and other measurements of health seeking behaviors or self-reported health status are routinely collected via questionnaires with the implicit understanding that although such measurements are often well-validated instruments, they should be viewed as proxies of the latent factors inducing confounding in causal queries about potential public health or public policy interventions on health and related outcomes. Extensive discussion of proxies encountered in health and social sciences can be found in \cite{lipsitch2010negative,kuroki2014measurement,miao2018identifying,sofer2016negative,Shi2019MultiplyRC,shi2020selective}; the proposed proximal causal inference framework is therefore a unifying framework for identification and inference leveraging the various types of proxies that have appeared in prior literature.

Other prominent examples of treatment and outcome proxies are negative control treatment and outcome variables routinely used in recent literature on the protective effectiveness of COVID-19 vaccination in real world settings \citep{patel2020postlicensure,dagan2021bnt162b2,thompson2021effectiveness,olson2022effectiveness,li2022double}. For example, in \cite{li2022double}, immunization visits before December 2020, when the COVID-19 vaccine became available, were used as treatment confounding proxies and the following diagnoses after April 5, 2021 were used as outcome confounding proxies: arm/leg cellulitis, eye/ear disorder, gastroesophageal disease, atopic dermatitis, and injuries. Aside for these proxies, there may also be factors that
can accurately be described as true common causes of treatment and outcome
processes; these variables collected in bucket type 1) may in fact
include age, gender, race or ethnicity, and years of education depending on the context. Thus, rather than assuming that exchangeability can be attained by adjusting for measured
covariates, our proposed proximal framework requires that the investigator can correctly select proxies of types 2) and 3) relative to a latent factor (possibly multivariate) that would in principle suffice to account for confounding; this condition is formalized using the potential outcomes framework in the following section.               
In terms of proximal estimation and inference, \cite{tchetgen2020} focus primarily on so-called proximal g-computation, a generalization of Robins' g-computation algorithm which may be viewed essentially as maximum likelihood estimation, requiring a correctly specified model for the entire data generating mechanism; in case of linear models, they proposed a proximal recursive two-stage least squares algorithm for point and time-varying treatments which is somewhat more robust provided the specified linear models are correct.

In this paper, we develop a general semiparametric framework for proximal causal inference about the population average treatment effect (ATE) and the average treatment effect on the treated (ATT) using proxies of types 2) and 3) while accounting for a possibly large number of observed covariates in the point treatment setting. In addition, we establish an alternative condition to that of \cite{miao2018identifying}, \cite{Shi2019MultiplyRC} and \cite{tchetgen2020} for nonparametric proximal identification of the average treatment effect in case of a point intervention;
and we subsequently characterize the semiparametric efficiency bound for the identifying functional of the ATE under two key semiparametric models that place different restrictions on the observed data distribution, as well as under a nonparametric model in which both sets of restrictions are relaxed. 
We then propose a class of doubly robust locally efficient estimators of the average treatment effect that are consistent provided one of two aforementioned models restricting the observed data distribution is correct, but not necessarily both. The proposed estimators are locally efficient in the sense that when all working models are correctly specified (i.e., the intersection submodel), our estimators achieve the semiparametric efficiency bound for estimating the average treatment effect under the union model. Analogous results are obtained for the ATT in the Appendix~\ref{sec:G}. 

The remainder of the article is organized as followed. In Section~\ref{sec:identification},  we briefly review proximal identification results of ATE \citep{miao2018identifying,Shi2019MultiplyRC,tchetgen2020}  before introducing an alternative condition for nonparametric identification of ATE (and ATT).  In Section~\ref{sec:semi}, we develop semiparametric theory for proximal causal inference \citep{tchetgen2020} and describe construction of doubly robust and semiparametric locally efficient estimators.
In Section~\ref{sec:connection}, we draw parallels between the proposed doubly robust proximal estimators and standard augmented inverse-probability-weighted estimators developed by Robins and colleagues  under no unmeasured confounding \citep{Scharfstein1999}; we establish that the former may be viewed as a generalization of the latter allowing for unmeasured confounding under our identifying assumptions.
Simulation studies are presented in Section~\ref{sec:numeric}.
Section~\ref{sec:real} describes a real data application on evaluating the effectiveness of right heart catheterization in the intensive care unit of critically ill patients.
The article concludes with a discussion of future work in Section~\ref{sec:discussion}. Proofs and additional results are provided in the Appendix.

\section{Nonparametric proximal identification of the average treatment effect}\label{sec:identification}

\subsection{Background}

We consider estimating the effect of a binary treatment $A$ on an outcome $Y$ subject to potential unmeasured confounding. Throughout, we let $U$ denote the unmeasured confounder (possibly vector-valued) which may be discrete, or continuous, or include both types of variables. Let $Y(a)$, $a = 0, 1$ denote the potential outcome that would be observed if the treatment were set to $a$. We are interested in
the population average treatment effect defined as $\psi = \E[Y (1) - Y (0)]$. We assume that the following consistency assumption holds: 
\begin{assumption}\emph{(Consistency)}\label{asm:consistency}
$Y = Y(A)$ almost surely.
\end{assumption}
Moreover, suppose that one has measured covariates $L$ such that:
\begin{assumption}\emph{(Positivity)}\label{asm:positivity}
$0<\Pr(A=a|L)<1$ almost surely, $a=0,1$.
\end{assumption}

 A common identification strategy in observational studies invokes exchangeability condition on the basis of measured covariates.
\begin{assumption}\emph{(Exchangeability)}\label{asm:exchangeability}
$Y(a)\perp A|L$ for $a=0,1$. 
\end{assumption}
Assumption~\ref{asm:exchangeability} is sometimes interpreted as stating that $L$ includes all common
causes of $A$ and $Y;$ an assumption represented in causal directed acyclic
graph (DAG) in Figure~\ref{fig:sra dag0}(a), in which \yifan{case} $L$ is of type 1). Under Assumptions~\ref{asm:consistency}-\ref{asm:exchangeability}, it is known that the counterfactual mean $\E[Y(a)]$ is identified by celebrated g-formula \citep{robins1986new,hernan2020causal}.

It is also interesting to consider alternative data generating mechanisms under which Assumption~\ref{asm:exchangeability} holds, illustrated in Figures~\ref{fig:sra dag0}(b) and \ref{fig:sra dag0}(c), with the first of type 2) where $L$ includes all causes of $A$ that share an unmeasured common cause $U$ (and therefore are associated) with $Y$; while the second is of type 3) where $L$ includes all
causes of $Y$ that share an unmeasured common cause $U$ (and therefore are associated) with $A$. 
Measured covariates of types 1), 2), and 3) may coexist, as depicted in Figure~\ref{fig:sra dag0}(d), in which \yifan{case} $L$ has been decomposed into three bucket types of measured covariates $L=(X,W,Z),$ such that $X$ are measured covariates of type 1), $Z$ are measured covariates of type 2), and $W$ are measured covariates of type 3).
All four settings represented in Figure~\ref{fig:sra dag0} illustrate possible data generating mechanisms under which exchangeability assumption~\ref{asm:exchangeability} holds, without necessarily requiring that the analyst identify which bucket type each covariate in $L$ belongs to. Note that all four settings rule out the presence of an unmeasured common cause of $A$ and $Y$, therefore ruling out unmeasured confounding.

\vspace{1cm}
\begin{figure}[h]
  \centering
  \vfill
  \resizebox{410pt}{!}{
    \begin{tikzpicture}[state/.style={circle, draw, minimum size=0.7cm}]
  \def\Ax{0}
  \def\Ay{0}
  \def\offset{2.5}
  \def\Bx{\Ax+5}
  \def\By{\Ay}

  \node[state,shape=circle,draw=black] (Y) at (\Bx,\By) {$Y$};
  \node[state,shape=circle,draw=black] (A) at (\Bx-2.5,\By) {$A$};
  \node[state,shape=circle,draw=black] (L) at (\Bx-1.25,\By+1.5) {$L$};

  \draw [-latex] (L) to [bend left=0] (A);
  \draw [-latex] (L) to [bend left=0] (Y);
  \draw [-latex] (A) to [bend left=0] (Y);

\end{tikzpicture} 
    \begin{tikzpicture}[state/.style={circle, draw, minimum size=0.7cm}]
  \def\Ax{0}
  \def\Ay{0}
  \def\offset{2.5}
  \def\Bx{\Ax+5}
  \def\By{\Ay}
  
  \node[state,shape=circle,draw=black] (U) at (\Bx,\By+1.5) {$U$};
  \node[state,shape=circle,draw=black] (Y) at (\Bx,\By) {$Y$};
  \node[state,shape=circle,draw=black] (A) at (\Bx-2.5,\By) {$A$};
  \node[state,shape=circle,draw=black] (L) at (\Bx-1.25,\By+1.5) {$L$};

  \draw [-latex] (L) to [bend left=0] (A);
  \draw [-latex] (U) to [bend left=0] (L);
  \draw [-latex] (A) to [bend left=0] (Y);
\draw [-latex] (U) to [bend left=0] (Y);

\end{tikzpicture}
    \begin{tikzpicture}[state/.style={circle, draw, minimum size=0.7cm}]
  \def\Ax{0}
  \def\Ay{0}
  \def\offset{2.5}
  \def\Bx{\Ax+5}
  \def\By{\Ay}
  
  \node[state,shape=circle,draw=black] (U) at (\Bx-2.5,\By+1.5) {$U$};
  \node[state,shape=circle,draw=black] (Y) at (\Bx,\By) {$Y$};
  \node[state,shape=circle,draw=black] (A) at (\Bx-2.5,\By) {$A$};
  \node[state,shape=circle,draw=black] (L) at (\Bx-1.25,\By+1.5) {$L$};

  \draw [-latex] (L) to [bend left=0] (Y);
  \draw [-latex] (U) to [bend left=0] (L);
  \draw [-latex] (A) to [bend left=0] (Y);
\draw [-latex] (U) to [bend left=0] (A);

\end{tikzpicture}
    \begin{tikzpicture}[state/.style={circle, draw, minimum size=0.7cm}]
  \def\Ax{0}
  \def\Ay{0}
  \def\offset{2.5}
  \def\Bx{\Ax+5}
  \def\By{\Ay}
  \node[state,shape=circle,draw=black] (Z) at (\Bx-4,\Ay+1.5) {$Z$};
  \node[state,shape=circle,draw=black] (U1) at (\Bx-2.5,\Ay+1.5) {$U_1$};
  \node[state,shape=circle,draw=black] (U2) at (\Bx,\By+1.5) {$U_2$};
  \node[state,shape=circle,draw=black] (Y) at (\Bx,\By) {$Y$};
  \node[state,shape=circle,draw=black] (A) at (\Bx-2.5,\By) {$A$};
  \node[state,shape=circle,draw=black] (X) at (\Bx-1.25,\By+1.5) {$X$};
  \node[state,shape=circle,draw=black] (W) at (\Bx+1.5,\Ay+1.5) {$W$};

  \draw [-latex] (A) to [bend left=0] (Y);
  \draw [-latex] (X) to [bend left=0] (A);
  \draw [-latex] (X) to [bend left=0] (Y);
  \draw [-latex] (Z) to [bend left=0] (A);
  \draw [-latex] (W) to [bend left=0] (Y);

  \draw [-latex] (U1) to [bend left=0] (Y);
  \draw [-latex] (U2) to [bend left=0] (A);
    \draw [-latex] (U1) to [bend left=0] (Z);
  \draw [-latex] (U2) to [bend left=0] (W);

\end{tikzpicture}
}\\
  (a) type 1);~~~~~ (b) type 2);~~~~~ (c) type 3);~~~~~ (d) Coexistence of 1) 2) 3)
  \caption{DAGs representing treatment and outcome confounding proxies when exchangeability holds.}
  \label{fig:sra dag0}
\end{figure}
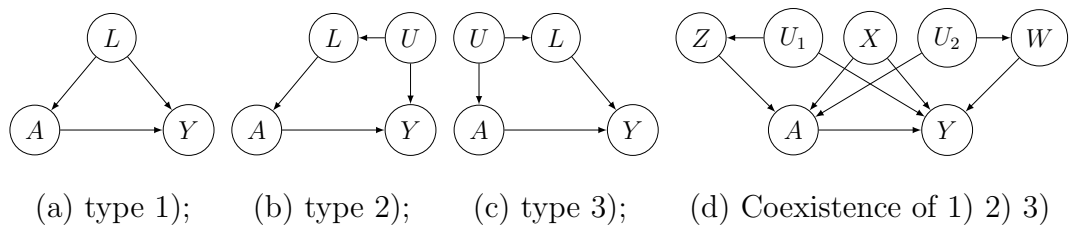
\vspace{0.5cm}

In order for exchangeability to hold in Figure~\ref{fig:sra dag0}(d), it
must be that as encoded in the DAG, unmeasured variables $U_{1}$ and $U_{2}$
are independent conditional on $A,X,Z,$ and $W$; otherwise, as illustrated in
Figure~\ref{fig:sra dag3}, the unblocked backdoor path $A-U_{2}-U_{3}-U_{1}-Y$ would
invalidate Assumption~\ref{asm:exchangeability}. As shown in
the next section, it is possible to relax this conditional independence assumption and therefore Assumption~\ref{asm:exchangeability} while preserving identification of the counterfactual mean parameter despite the presence of unmeasured confounding.

\vspace{1cm}
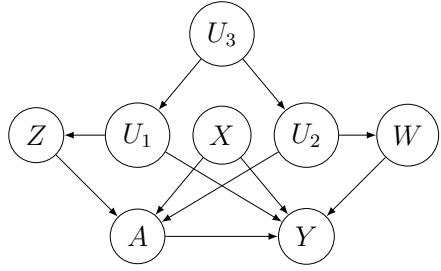
\begin{figure}[h]
  \centering
  \vfill
  \resizebox{170pt}{!}{%
    \begin{tikzpicture}[state/.style={circle, draw, minimum size=0.6cm}]
  \def\Ax{0}
  \def\Ay{0}
  \def\offset{2.5}
  \def\Bx{\Ax+5}
  \def\By{\Ay}
  \node[state,shape=circle,draw=black] (Z) at (\Bx-4,\Ay+1.5) {$Z$};
  \node[state,shape=circle,draw=black] (U1) at (\Bx-2.5,\Ay+1.5) {$U_1$};
  \node[state,shape=circle,draw=black] (U2) at (\Bx,\By+1.5) {$U_2$};
  \node[state,shape=circle,draw=black] (Y) at (\Bx,\By) {$Y$};
  \node[state,shape=circle,draw=black] (A) at (\Bx-2.5,\By) {$A$};
  \node[state,shape=circle,draw=black] (X) at (\Bx-1.25,\By+1.5) {$X$};
    \node[state,shape=circle,draw=black] (U3) at (\Bx-1.25,\By+3) {$U_3$};
  \node[state,shape=circle,draw=black] (W) at (\Bx+1.5,\Ay+1.5) {$W$};

  \draw [-latex] (A) to [bend left=0] (Y);
  \draw [-latex] (X) to [bend left=0] (A);
  \draw [-latex] (X) to [bend left=0] (Y);
  \draw [-latex] (Z) to [bend left=0] (A);
  \draw [-latex] (W) to [bend left=0] (Y);

  \draw [-latex] (U3) to [bend left=0] (U2);
  \draw [-latex] (U3) to [bend left=0] (U1);
  \draw [-latex] (U1) to [bend left=0] (Y);
  \draw [-latex] (U2) to [bend left=0] (A);
    \draw [-latex] (U1) to [bend left=0] (Z);
  \draw [-latex] (U2) to [bend left=0] (W);

\end{tikzpicture}
  }
  \vfill
  \caption{Coexistence of types 1), 2), and 3) proxies when exchangeability fails.}
  \label{fig:sra dag3}
\end{figure}
\vspace{0.5cm}

\subsection{Proximal identification}

Now consider a setting in which exchangeability assumption~\ref{asm:exchangeability} fails.
Suppose that one has partitioned $L$ into variables $(X,Z,W)$, such that $Z$ includes treatment-inducing confounding proxies, and $W$ includes outcome-inducing confounding proxies known to satisfy the following Assumptions~\ref{asm:condindy}-\ref{asm:condrand} which formalize the notion of covariate types 2) and 3).

\vspace{1cm}
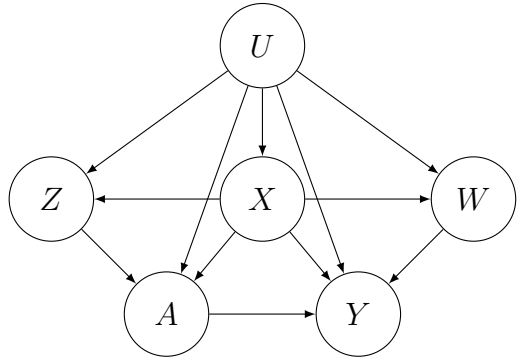
\begin{figure}[h]
  \centering
  \vfill
  \resizebox{200pt}{!}{%
    \begin{tikzpicture}[state/.style={circle, draw, minimum size=1.1cm}]
  \def\Ax{0}
  \def\Ay{0}
  \def\offset{2.5}
  \def\Bx{\Ax+5}
  \def\By{\Ay}
  \node[state,shape=circle,draw=black] (Z) at (\Bx-4,\Ay+1.5) {$Z$};
  \node[state,shape=circle,draw=black] (Y) at (\Bx,\By) {$Y$};
  \node[state,shape=circle,draw=black] (A) at (\Bx-2.5,\By) {$A$};
  \node[state,shape=circle,draw=black] (X) at (\Bx-1.25,\By+1.5) {$X$};
    \node[state,shape=circle,draw=black] (U) at (\Bx-1.25,\By+3.5) {$U$};
  \node[state,shape=circle,draw=black] (W) at (\Bx+1.5,\Ay+1.5) {$W$};

  \draw [-latex] (X) to [bend left=0] (W);
  \draw [-latex] (A) to [bend left=0] (Y);
  \draw [-latex] (X) to [bend left=0] (A);
  \draw [-latex] (X) to [bend left=0] (Z);
  \draw [-latex] (X) to [bend left=0] (Y);
  \draw [-latex] (Z) to [bend left=0] (A);
  \draw [-latex] (W) to [bend left=0] (Y);

  \draw [-latex] (U) to [bend left=0] (A);
  \draw [-latex] (U) to [bend left=0] (Z);
  \draw [-latex] (U) to [bend left=0] (Y);
  \draw [-latex] (U) to [bend left=0] (W);
  \draw [-latex] (U) to [bend left=0] (X);

\end{tikzpicture}
  }
  \vfill
  \caption{A causal DAG of proximal causal inference.}
  \label{fig:sra dag}
\end{figure}
\vspace{0.5cm}

\begin{assumption}\emph{(Conditional independence for $Y$)}\label{asm:condindy}
$Y \perp Z | U, X, A$.
\end{assumption}

\begin{assumption}\emph{(Conditional independence for $W$)}\label{asm:condindw}
$W \perp (Z,A) | U, X$.
\end{assumption}

\begin{assumption}\emph{(Conditional randomization)}\label{asm:condrand}
\begin{align*}
Y (a) \perp A|U,X ~\text{for}~ a=0,1.
\end{align*}
\end{assumption}

Figure~\ref{fig:sra dag} provides a graphical representation of Assumptions~\ref{asm:condindy}-\ref{asm:condrand}. For instance,  Assumptions~\ref{asm:condindy} and \ref{asm:condindw} formalize the hypothesis in \cite{li2022double} discussed in Section~\ref{sec:intro} that prior immunization visits are valid treatment confounding proxies and diagnoses such as arm/leg cellulitis, eye/ear disorder, injuries are valid outcome confounding proxies.
As discussed in \cite{shi2020selective}, a number of alternative DAGs are in fact compatible with Assumptions~\ref{asm:condindy}-\ref{asm:condrand}; see Table~A.1 of \cite{tchetgen2020} and \cite{shi2020selective}.
In addition, we make the following positivity assumption for nonparametric identification. 

\begin{assumption}\emph{(Positivity)}\label{asm:positivity2}
$0<\Pr(A=a|U,X)<1$ almost surely, $a=0,1$.
\end{assumption}

\yf{Assumption~\ref{asm:positivity2} essentially states that the probability of having a particular level of exposure, conditional on $X$ and $U$, is greater than zero for all strata.}
In order to identify the population average treatment effect, \cite{miao2018identifying,tchetgen2020} consider the following assumptions:

\begin{assumption}\emph{(Completeness)}\label{asm:completeness}
For any square-integrable function $g$ and for any $a,x$, $\E[g(U)|Z, A=a, X=x]=0$ almost surely if and only if $g(U)=0$ almost surely.
\end{assumption}
These conditions are formally known as completeness conditions which
can accommodate both categorical and continuous confounders. Completeness is a technical condition central to the study of sufficiency in foundational theory of
statistical inference. 
\yifan{We note that the completeness assumption~\ref{asm:completeness} rules out conditional independence of $U$ and $Z$ given $A$ and $X$.}
Here one may interpret the completeness condition as a requirement relating the range of $U$ to that of $Z$ which essentially states that the set of proxies
must have sufficient variability relative to variability of $U$. 
The condition is easiest understood in the case of categorical $U$ and $Z$ with number
of categories $d_{u}$ and $d_{z}$, respectively$.$ In this case,
completeness assumption~\ref{asm:completeness} requires that
\begin{equation}
d_{z}  \geq d_{u},\label{categorical completeness}
\end{equation}
which states that $Z$ must have at least as many categories as
$U$. Intuitively, condition $\left(  \ref{categorical completeness}\right)$
states that proximal causal inference can potentially account for unmeasured
confounding in the categorical case as long as the number of categories of $U$
is no larger than that of $Z$ (\cite{miao2018identifying}, \cite{Shi2019MultiplyRC} and \cite{tchetgen2020}). This further
provides a rationale for measuring a rich set of baseline characteristics in
observational studies as a potential strategy for mitigating unmeasured
confounding via the proximal approach we describe below. 
Many commonly used parametric and semiparametric models such as exponential families \citep{newey2003instrumental} satisfy the completeness condition. For nonparametric regression models, results of \cite{d2011completeness} and \cite{darolles2011nonparametric} can be used to justify the completeness condition, although they focused on nonparametric instrumental variable problems where completeness also plays an important role.
\yifan{\cite{chen2014local} and \cite{andrews2017examples} showed that if $Z$ and $U$ are continuously distributed and the dimension of $Z$ is larger than that of $U$, then under a mild regularity condition, the
completeness condition holds generically in the sense that the set of distributions for which completeness fails has a property similar to  being essentially Lebesgue measure zero. More specifically, as formally shown by \cite{canay2013testability}, distributions
for which a completeness condition fails can approximate distributions for which it holds arbitrarily well in the total variation distance sense. Thus, while completeness conditions may themselves not be directly testable, one may argue as in \cite{canay2013testability} that they are commonly satisfied.} 
We further refer to \cite{chen2014local}, \cite{andrews2017examples}, and Section~2 of the Supplementary Material of \cite{miao2022} for a review and examples of completeness. 

\cite{miao2018identifying} established the following nonparametric identification result which we have adapted to the proximal inference setting. 
\begin{theorem}\label{thm:id}\citep{miao2018identifying}
Suppose that there exists an outcome confounding bridge function $h(w, a, x)$ that solves the following integral equation
\begin{align}\label{eq:condexpec}
\E[Y |Z,A,X] = \int h(w, A, X) dF(w|Z,A,X),
\end{align}
almost surely.
\begin{itemize}
\item[(1)] (Factuals) Under Assumptions~\ref{asm:condindy}, \ref{asm:condindw}, \ref{asm:positivity2},
and \ref{asm:completeness}, one has that
\begin{align}\label{eq:u1}
\E[Y |U,A,X] = \int h(w, A, X) dF(w|U,X),
\end{align}
almost surely.
\item[(2)] (Causal) Suppose that Equation~\eqref{eq:u1} holds. Under Assumptions~\ref{asm:consistency}, \ref{asm:condrand}, and \ref{asm:positivity2}, the counterfactual mean $\E[Y(a)]$ is nonparametrically identified by
\begin{align}\label{eq:proximalg1}
\E[Y (a)] = \int_{\mathcal X} \int h(w, a, x)dF(w|x)dF(x),
\end{align}
and thus the average treatment effect is identified by $$\psi = \int_{\mathcal X} \int [h(w, 1, x)- h(w, 0, x)]dF(w|x)dF(x).$$
\end{itemize}
\end{theorem}

\begin{remark}
Technical conditions for the existence of a solution to Equation~\eqref{eq:condexpec} are provided in the Appendix~\ref{sec:existence}. In particular, we note that the following assumption: 
\begin{assumption}\emph{(Completeness)}\label{asm:completeness1.2}
For any square-integrable function $g$ and for any $a,x$, $\E[g(Z)|W, A=a, X=x]=0$ almost surely if and only if $g(Z)=0$ almost surely.
\end{assumption}
\noindent along with the regularity conditions on the singular value decomposition
of the conditional mean operators together, suffice for the existence of a solution to Equation~\eqref{eq:condexpec}.
\end{remark}

Theorem~\ref{thm:id} serves as a basis for inference in \cite{Shi2019MultiplyRC,tchetgen2020,miao2018confounding}, as described in Section~\ref{sec:semi}.
Equation~\eqref{eq:condexpec} defines a so-called inverse problem known as a Fredholm integral equation of the first kind \citep{kress1989linear,miao2018identifying}. 
 Importantly, while the theorem does not require uniqueness of a solution to the integral equation~\eqref{eq:condexpec}, all solutions lead to a unique value of the ATE. 
Note that Equation~\eqref{eq:u1} highlights the inverse problem nature of the task
accomplished by Equation~\eqref{eq:proximalg1} (\cite{tchetgen2020} refer to \eqref{eq:proximalg1} as proximal g-formula), which is to determine a function $h$ that satisfies this equality without
explicitly modeling or estimating the latent factor $U$. A remarkable feature of proximal causal inference 
is that accounting for $U$ without either measuring $U$ directly or estimating its distribution can
be accomplished provided that the set of proxies though imperfect, is sufficiently rich so that the
inverse problem admits a solution.

\begin{remark}\label{remark1}
\yifan{There are different strategies to achieve identification in proximal causal inference.}
Instead of taking Equation~\eqref{eq:condexpec} as a starting point,
\cite{miao2018confounding} consider an alternative identifying condition in which  they assume that there exists a bridge function \yifan{$h'(w, a, x)$} such that \yifan{$\E[Y|U,a,X] = \int h' (w, a, X) dF(w|U, X)$} almost surely. In addition, \cite{miao2018confounding} replace completeness assumption~\ref{asm:completeness} which involves unobservables with an alternative completeness condition that only involves observables; that is, for any square-integrable function $g$ and for any $a,x$, $\E [g(W)|Z, a, x] = 0$ almost surely if and only if $g(W) = 0$ almost surely. They then establish that such function \yifan{$h'$} must solve Equation~\eqref{eq:condexpec}, and arrive at the very same identifying proximal g-formula~\eqref{eq:proximalg1}. 
\end{remark}

\subsection{A new proximal identification result}

In this section, we establish an alternative proximal identification result to that of \cite{miao2018identifying}. We first consider identification with a discrete variable $U$. Suppose that $W,Z,U$ are discrete variables, each with $d$ categories. 
For notational convenience, we write $P(Z|w) = (\Pr(z_1|w), \ldots,\Pr(z_d|w))^T$, $P(z|W) = (\Pr(z|w_1), \ldots,\Pr(z|w_d))$, and $P(Z|W) = (\Pr(Z|w_1), \ldots,\Pr(Z|w_d))$ to denote a column vector, a row vector, and a matrix, respectively. 
For other variables, vectors and matrices are analogously defined. 
By $(Z,A)\perp W|U,X$, we have that
\begin{align*}
P(Z|W,a,x) & = P(Z|U,a,x)P(U|W,a,x),\\
P(a|W,x)^{-1} & = P(a|U,x)^{-1} P(U|W,a,x),
\end{align*}
where $P(a|W,x)^{-1}$ denotes $(1/\Pr(a|w_1,x), \ldots,1/\Pr(a|w_d,x))$ and $P(a|U,x)^{-1}$ denotes \\
$(1/\Pr(a|u_1,x), \ldots,1/\Pr(a|u_d,x))$. Therefore, assuming that $P(Z|W,a,X)$ is invertible for any $a=0,1$ and $x$, we have that
\begin{align*}
P(a|W,x)^{-1}P^{-1}(Z|W,a,x) P(Z|U,a,x) & = P(a|U,x)^{-1},
\end{align*}
where $P^{-1}(Z|W,a,x)$ denotes the inverse of the matrix $P(Z|W,a,x)$.
Furthermore, by $Y \perp Z|U, A, X$,
\begin{align*}
P(y|U,a,x) 
& =   [P(a|W,x)^{-1} P^{-1}(Z|W,a,x) P(Z|U,a,x)] \times P(a|U,x) \times P(y|U,a,x) \\
& = P(a|W,x)^{-1} P^{-1}(Z|W,a,x)  P(y,a,Z|U,x),
\end{align*}
where $\times$ denotes element-wise multiplication.
Upon multiplying both sides by $(\Pr(u_1|x),\ldots,\Pr(u_d|x))^T$, we have that
\begin{align*}
\Pr(y(a)|x)  = P(a|W,x)^{-1} P^{-1}(Z|W,a,x)P(y,a,Z|x).
\end{align*}
Therefore establishing identification of $\Pr(y(a)|x)$, and thus, identification of corresponding functionals, such as counterfactual means and average treatment effect. 
Next, we extend this identification result to a general setting in which $U$ can include both categorical and continuous factors, under the following completeness condition.

\begin{assumption}\emph{(Completeness)}\label{asm:completeness2}
For any square-integrable function $g$ and for any $a,x$, $\E[g(U)|W, A=a, X=x]=0$ almost surely if and only if $g(U)=0$ almost surely.
\end{assumption}

Assumption~\ref{asm:completeness2} essentially states that $W$ must have sufficient variability relative to variability of $U$.

\begin{theorem}\label{thm:identification}
Suppose that there exists a treatment confounding bridge function $q(z,a,x)$ that solves the integral equation:
\begin{align}\label{eq:condexpect2}
\E[q(Z,a,X)|W,A=a, X] = \frac{1}{f(A=a|W, X)},
\end{align}
almost surely.
\begin{itemize}
\item[(1)] (Factuals) Under Assumptions~\ref{asm:condindw},  \ref{asm:positivity2}, and \ref{asm:completeness2}, one has that
\begin{align}\label{eq:u2}
\int q(z,a,X)dF(z|U,A=a,X) = \frac{1}{f(A=a|U,X)},
\end{align}
almost surely.
Furthermore, if Assumption~\ref{asm:condindy} also holds, one also has that
\begin{align}\label{eq:u3}
\E[Y|U,A=a,X=x] = \E[I(A=a)Yq(Z,a,X)|U,X=x] ,
\end{align}
almost surely.
\item[(2)] (Causal) Suppose that Equation~\eqref{eq:u3} holds. Under Assumptions~\ref{asm:consistency}, \ref{asm:condrand}, and \ref{asm:positivity2}, $\E[Y(a)]$ is nonparametrically identified by
\begin{align}\label{eq:proximalg2}
\E[Y (a)] =\int_{\cX}\int I(\tilde a =a) q(z,a,x)y dF(y,z,\tilde a|x) dF(x).
\end{align}
Therefore, the average treatment effect is identified by $$\psi = \int_{\mathcal X} \int  (-1)^{1-a} q(z,a,x) y dF(y,z,a|x) dF(x).$$
\end{itemize}
\end{theorem}

\yf{\begin{remark}
We note that Theorem~\ref{thm:identification} also applies for a continuous possibly multivariate exposure $A$ with
$$\E[Y (a)]  = \int_{\mathcal X}\int_{\mathcal Y,\mathcal Z}  q(z,a,x) y dF(y,z,a|x) dF(x).$$
\end{remark}
}

\begin{remark}
Formal technical conditions for the existence of a solution to \eqref{eq:condexpect2} are provided in the Appendix~\ref{sec:existence}. In particular, we note that the following assumption: 

\begin{assumption}\emph{(Completeness)}\label{asm:completeness2.2}
For any square-integrable function $g$ and for any $a,x$, $\E[g(W)|Z, A=a, X=x]=0$ almost surely if and only if $g(W)=0$ almost surely.
\end{assumption}
\noindent along with the regularity conditions on the singular value decomposition
of the conditional mean operators together, suffice for the existence of a solution to Equation~\eqref{eq:condexpect2}.
\end{remark}

Theorem~\ref{thm:identification} provides a new proximal identification result which complements that of \cite{miao2018identifying,tchetgen2020}. 
Similar to Equation~\eqref{eq:condexpec}, 
Equation~\eqref{eq:condexpect2} also defines a Fredholm integral equation of the first kind. 
Fredholm equations are well known to often be ill-posed and solving them requires a form of regularization in practice. In the next section, we will consider using semiparametric models as an implicit form of regularization. 
We note that exchangeability assumption, \yf{albeit a structural assumption}, can also be viewed as a form of regularization of Equations~\eqref{eq:condexpec} and \eqref{eq:condexpect2} which automatically yields unique and stable solutions to the integral equations. Intuitively, 
\yf{suppose that $U=\emptyset$, then Equation~\eqref{eq:condexpec} reduces to}
\begin{align*}
\E[Y |A,X] = \int \E(Y|w, A, X) dF(w|A,X),
\end{align*}
i.e., $h (w, a, x) = \E(Y |w, a, x)$; alternatively, \yf{suppose that $U=\emptyset$, then Equation~\eqref{eq:condexpect2} reduces to}
\begin{align*}
\frac{1}{f(A=a|X)} = \int \frac{1}{f(A=a|z, X)} dF(z|A=a, X),
\end{align*}
i.e., $q (z, a, x) = 1/f(A=a|z, x)$.

Note that again Equation~\eqref{eq:u2} highlights the inverse problem nature of the task
accomplished by Equation~\eqref{eq:proximalg2}, which is to determine a function $q$ that satisfies the equation without
explicitly modeling or estimating the latent factor $U$. The theorem reveals that accounting for $U$ without either measuring $U$ directly or estimating its distribution can
be accomplished provided that the set of proxies is sufficiently rich so that the inverse problem admits a solution. 

\begin{remark}
Note that similar to Remark~\ref{remark1}, instead of assuming Equation~\eqref{eq:condexpect2}, one could alternatively consider the following identifying condition: suppose that there exists a bridge function $q'(z, a, x)$ such that
${1}/{f(A=a|U,X)} = \int q'(z,a,X)dF(z|U,A=a,X)$ almost surely, and replace completeness assumption~\ref{asm:completeness2} with: for any square-integrable function $g$ and for any $a,x$, $\E [g(Z)|W, a, x] = 0$ almost surely if and only if $g(Z) = 0$ almost surely; subsequently by Theorem~\ref{ap:thm} presented in the Appendix~\ref{sec:C}, any solution $q'$ must also solve Equation~\eqref{eq:condexpect2}; furthermore identifying formula~\eqref{eq:proximalg2} still applies. 
\end{remark}

\begin{remark}
Identification of $q$ under the conditions of Theorem~\ref{thm:identification} offers an opportunity for identification of any smooth functional $\beta$ of the marginal counterfactual distribution that can be defined as a solution to a moment equation, say $\E(\int m(Y(a),a,X;\beta)d\mu(a))=0$, with $\mu$ a dominating measure of $a$. The theorem implies the following observed data moment equation analog obtained by reweighting the moment equation by the treatment confounding bridge function, i.e., $\E(q(Z,A,X)m(Y,A,X;\beta))=0$, provided that the expectation is well defined. For instance, the marginal distribution of $Y(a)$ at $y$ can be identified by $m(Y(a),a,x;y,\beta)=I(Y(a) \leq y) -\beta$. Therefore, identification Theorem~\ref{thm:identification} is more general than Theorem~\ref{thm:id} which only identifies the ATE.   
\end{remark}

\section{Semiparametric theory and inference}\label{sec:semi}

\yifan{We consider inference for $\psi$ under the semiparametric model $\cM_{sp}$
which places no restriction on the observed data distribution other than existence (but not necessarily uniqueness) of a bridge function $h$ that solves Equation~\eqref{eq:condexpec}.
Let $T:L_2(W,A,X)\rightarrow L_2(Z,A,X)$ be the conditional expectation operator given by $T(g)\equiv
\E[g(W,A,X)|Z,A,X]$, and the adjoint $T':L_2(Z,A,X)\rightarrow L_2(W,A,X)$ be $T'(g)\equiv \E[g(Z,A,X)|W,A,X]$. 
We consider the following regularity condition under the model:}
\begin{assumption}\label{asm:bound}
$T$ and $T'$ are surjective. 
\end{assumption}

\begin{theorem}\label{thm:eif} We have the following results.
\begin{itemize}
\item[(1)] Under Assumptions~\ref{asm:completeness1.2} and \ref{asm:completeness2.2}, $h$ and $q$ that solve integral equations~\eqref{eq:condexpec} and \eqref{eq:condexpect2} are uniquely identified.
\item[(2)]
The efficient influence function of $\psi$ under the semiparametric model $\cM_{sp}$ evaluated at the submodel where Assumption~\ref{asm:bound} holds, Equation~\eqref{eq:condexpect2} holds at the true data generating law, and $h$, $q$ are uniquely defined, is given by
\begin{align}\label{eq:eif}
EIF(\psi) = (-1)^{1-A} q(Z,A,X) [Y- h(W,A,X)] + h(W,1,X) - h(W,0,X) - \psi.
\end{align}
Therefore, the corresponding semiparametric local efficiency bound of $\psi$ 
equals $\E[EIF^2(\psi)]$.
\item[(3)] Equation~\eqref{eq:eif} admits the double robustness property: $\E[EIF(q^{*}, h^{*};\psi)]=0$ provided that either $h^{*}$ is a solution to equation \eqref{eq:condexpec}, or $ q^{*}$ is a solution to equation \eqref{eq:condexpect2}, but both do not necessarily hold. 
\end{itemize}
\end{theorem}

\yifan{We note that at the submodel where both Assumptions~\ref{asm:completeness1.2} and \ref{asm:completeness2.2} hold, both completeness of the law of $W$ given $Z$, $A = a$, $X=x$ and of the law of $Z$ given $W$, $A = a$, $X=x$ must hold. Therefore, for
$W$ and $Z$ finitely valued, this imposes the restriction that the
sample space of $Z$ and $W$ have equal cardinality; likewise, for continuous $W$ and $Z$, this requires that $Z$ and $W$ are of equal dimension. In the event that available candidate proxies $W$ and $Z$ are of unequal dimensions, say $W$ has higher dimension than $Z$, it may be possible to coarsen $W$ so that its dimension matches the dimension of $Z$ without compromising identification provided $Z$ has higher dimension than $U$; although a formal approach to operationalize such coarsening is currently lacking in the literature: see \cite{Shi2019MultiplyRC} for additional discussion in the categorical unmeasured confounding case. However, we note that formally, the double robustness result of Theorem~\ref{thm:eif}(3) does not strictly require either Assumption~\ref{asm:completeness1.2} or \ref{asm:completeness2.2} to hold, as uniqueness of either confounding bridge function is not necessary for the efficient influence function to be an unbiased moment equation, only that at least one of the bridge functions satisfies the corresponding integral equation at the true data generating law. For inference, in principle, one may wish for greater robustness to estimate $\psi$ under a nonparametric model for both nuisance functions $h$ and $q$; this is in fact the approach taken by \cite{ghassami2022minimax} and \cite{kallus2021causal}, who recently adopted the efficient influence function~\eqref{eq:eif} based on an earlier preprint of the current paper, to develop an adversarial inference framework which accommodates reproducing kernel Hilbert space, neural networks and other nonparametric or machine learning estimators of $h$ and $q$ \yf{(see also \cite{singh2020kernel,mastouri2021proximal,kompa2022deep} who propose nonparametric methods to evaluate average treatment effect with proximal causal inference)}. In order for the resulting estimator of the causal effect to be regular, these works require that both nuisance functions can be estimated at rates faster than $n^{-1/4}$ which may not be feasible where $L$ is of moderate to high dimension or depending on the extent to which the integral equations defining either confounding bridge function are ill-posed.}
It is therefore of interest to develop a doubly robust estimation approach that a priori posits low-dimensional working models for $h$ and $q$, but however is guaranteed to deliver valid inferences about $\psi$
provided that one but not necessarily both low dimensional models used to estimate $h$ and $q$ can be specified correctly. In this paper, we focus primarily on developing such a doubly robust approach much in the spirit of \cite{Scharfstein1999}, and demonstrate its ability to resolve concerns about ill-posedness and partial model misspecification.  In order to describe our proposed doubly robust approach, consider the following two semiparametric models that place parametric restrictions on different components of the observed data
likelihood while allowing the rest of the likelihood to remain unrestricted:

\noindent $\mathcal M_1$: $h(W,A,X)$ is assumed to be correctly specified and suppose Assumptions~\ref{asm:consistency}, \ref{asm:condindy}-\ref{asm:completeness}, and \ref{asm:completeness2.2} hold;\\
$\mathcal M_2$: $q(Z,A,X)$ is assumed to be correctly specified and suppose Assumptions~\ref{asm:consistency}, \ref{asm:condindy}-\ref{asm:positivity2}, and  \ref{asm:completeness1.2}-\ref{asm:completeness2} hold.

Our proposed doubly robust locally efficient estimator thus entails modeling both $h$ and $q$, however, as we will show below only one of these models will ultimately need to be correct for valid inferences about the ATE, without knowing a priori which model is correct. Directly modeling the outcome and treatment confounding bridge functions is a simple and practical regularization strategy that obviates the need to solve complicated integral equations that are well-known to be ill-posed and therefore to admit unstable solutions in practice 
\yifan{\citep{ai2003efficient,newey2003instrumental,hall2005nonparametric,horowitz2011applied}}.
\yifan{Illposedness refers to the discontinuity of the operator $T^{-1}$ (or $T'^{-1}$), and therefore $T^{-1}\hat r$ (or $T'^{-1} \hat r$) might not converge to $T^{-1} r$ (or $T'^{-1} r$) under a given norm even if $\hat r$ converges to $r$ with respect to a given norm for $r$ in the range of $T$ (or $T'$).}
In the Appendix~\ref{sec:E}, we also consider semiparametric efficient inference in submodels $\mathcal M_1$ and $\mathcal M_2$, respectively.

 Interestingly, the second part of Theorem~\ref{thm:if} in the Appendix~\ref{sec:E} suggests that, although $q$ solves an integral Equation~\eqref{eq:condexpect2} involving the reciprocal of the propensity score function $1/f(A|W,X)$, surprisingly, inferences about a model for $q$ may be obtained without the need for estimating the propensity score provided that an influence function for $t$ is used as an estimating equation. In other words, influence function based estimation of $t$ is fully robust to misspecification of the propensity score \yifan{as it implicitly uses a nonparametric estimator of the propensity score.}
The derived influence functions in Theorem~\ref{thm:if} motivate various 
 estimating equations for the corresponding confounding bridge functions $h(W,A,X;b)$ and $q(Z,A,X;t)$, respectively. For instance, if $b$ and $t$ are of dimensions $p_x+p_z+2$ and $p_x+p_w+2$, a natural choice of estimating equations are
\begin{align}\label{est:1}
\PP_n \left\{[Y-h(W,A,X;b)](1,Z,A,X)^T\right\}=0,\\
\PP_n \left\{(-1)^{1-A}q(Z,A,X;t)(1,W,A,X)^T-(0,(0)_{p_w},1,(0)_{p_x})^T\right\}=0,\label{est:2}
\end{align}
which correspond to $m(Z,A,X)=(1,Z,A,X)^T$ and $n(W,A,X)=(-1)^{1-A}(1,W,A,X)^T$ defined in Theorem~\ref{thm:if},
where $p_w$ is the dimension of $W$, $p_z$ is the dimension of $Z$, $p_x$ is the dimension of $X$, and the corresponding estimators are denoted by $\widehat h$ and $\widehat q$, respectively. 
\yifan{Linearity in $W$ in $h$ is essentially implied by a proportional relationship between the confounding effects of $U$ on $Y$ and $W$, although it does not necessarily imply the latter; likewise, linearity in $Z$ in $q$ is implied by (but does not necessarily imply) a logit model between $A$ and $U$ under certain conditions about the distribution of $U$, e.g., Gaussian $U$.}
Such estimators can then be used to construct a corresponding substitution estimator of $\psi$.
Specifically, proximal outcome regression (POR) and proximal inverse probability weighted (PIPW) estimators are defined as
\begin{align*}
\widehat \psi_{POR} =& \PP_n \left\{\widehat h(W,1,X) -\widehat h (W,0,X)\right\},\\
\widehat \psi_{PIPW} =& \PP_n \left\{(-1)^{1-A} \widehat q(Z,A,X) Y\right\},
\end{align*}
respectively. Note that as established in the Appendix~\ref{sec:E}, construction of a locally efficient estimator of $h$ under $\cM_1$ and $q$ under $\cM_2$, requires correct specification of additional components of the observed data law beyond $q$ and $h$. In principle, a locally efficient estimator of $h$ under $\cM_1$ may then be used to construct a locally efficient estimator of $\psi$ under $\cM_1$ by the plug-in principle \yifan{\citep{bickel2003}}. However, as pointed out by \cite{stephens2014locally}, such additional modeling efforts seldom deliver the anticipated efficiency gain when, as in the current case, they involve complex features of the observed data distribution which are difficult to model correctly, and thus the potential prize of attempting to attain semiparametric local efficiency for $h$ and $q$ may not always be worth the chase. Furthermore, as our primary objective is to obtain doubly robust locally efficient inferences about $\psi$, we show next that such a goal can be attained without necessarily obtaining a locally efficient estimator of $h$ under $\cM_1$ and $q$ under $\cM_2$. For these reasons, optimal index functions $m_{eff}(Z,A,X)$ and $n_{eff}(W,A,X)$ given in the Appendix~\ref{sec:E} are not considered for implementation. Instead, the simpler, albeit inefficient, estimators $\widehat h$ and $\widehat q$ are used in construction of a doubly robust locally efficient estimator of $\psi$ given in the next theorem.
\begin{theorem}\label{thm:estimator}
Under standard regularity conditions given in the Appendix~\ref{sec:F}, 
\begin{align*}
\widehat \psi_{PDR} = \PP_n\left\{ (-1)^{1-A} \widehat q(Z,A,X) [Y- \widehat h(W,A,X)] +\widehat h(W,1,X) -\widehat h(W,0,X) \right\},
\end{align*}
is a consistent and asymptotically normal estimator of $\psi$ under the semiparametric union model $\mM_{union} =\mM_1 \cup \mM_2$. Furthermore,  $\widehat \psi_{PDR}$ is semiparametric locally efficient in $\cM_{sp}$ at the intersection submodel $\mM_{int} =\mM_1 \cap \mM_2$ where Assumption~\ref{asm:bound} also holds.
\end{theorem}

\section{Connection with AIPW estimator under exchangeability}\label{sec:connection}

It is interesting to relate the influence function of the ATE in the proximal inference framework to the standard augmented inverse-probability-weighted (AIPW) estimator of \cite{robinsetal1994} under exchangeability. In fact, \yf{suppose that $U=\emptyset$}, then Equation~\eqref{eq:condexpec} reduces to
\begin{align*}
h(W,a,X) = \E[Y|W,A=a,X],
\end{align*}
and Equation~\eqref{eq:condexpect2} reduces to
\begin{align*}
q(Z,a,X) = \frac{1}{f(A=a|Z, X)}.
\end{align*}
Therefore, the efficient influence function of $\psi$,
\begin{align*}
(-1)^{1-A} q(Z,A,X) [Y- h(W,A,X)] + h(W,1,X) - h(W,0,X) - \psi,
\end{align*}
becomes
\begin{align}\label{eq:eifnou}
 \frac{(-1)^{1-A}}{f(A|Z,X)} \{Y- \E[Y|W,A,X]\} + \E[Y|W,A=1,X] -\E[Y|W,A=0,X] - \psi.
\end{align}
Equation~\eqref{eq:eifnou} has the form of the efficient influence function of the ATE under excheangeabiltiy given by \cite{robinsetal1994}. Therefore, the AIPW estimator can be viewed as a special case of the proposed efficient influence function under exchangeability. In this vein, the proposed proximal doubly robust estimator can be viewed as a generalization of the standard doubly robust estimator \citep{robinsetal1994} to account for potential unmeasured confounding.

\section{Numerical experiments}\label{sec:numeric}

In this section, we report simulation studies comparing various estimators we have proposed under varying degree of model misspecification.

\subsection{Simulation setup}

We first describe the data generating mechanism. 
Covariates $X$ are generated from a multivariate normal distribution $N(\Gamma_x,\Sigma_x)$. We then generate $A$ conditional on $X$ from a Bernoulli distribution. 

Next, we generate $Z,W,U$ from the following multivariate normal distribution,
\[
\left( Z,W,U\right) |A,X\sim MVN\left( \left(
\begin{array}{c}
\alpha _{0}+\alpha _{a}A+\alpha _{x}X \\
\mu _{0}+\mu _{a}A+\mu _{x}X \\
\kappa _{0}+\kappa _{a}A+\kappa _{x}X
\end{array}
\right) ,\Sigma=\left(
\begin{array}{ccc}
\sigma _{z}^{2} & \sigma _{zw} & \sigma _{zu} \\
\sigma _{zw} & \sigma _{w}^{2} & \sigma _{wu} \\
\sigma _{zu} & \sigma _{wu} & \sigma _{u}^{2}
\end{array}
\right) \right).
\]

Finally, $Y$ is generated from $\E\left( Y|W,U,A,Z,X\right)$ plus a normal noise $N(0,\sigma_y^2)$ with
\begin{eqnarray*}
\E\left( Y|W,U,A,Z,X\right)  &=&\E\left( Y|U,A,Z,X\right) +\omega \left\{
W-\E\left( W|U,A,Z,X\right) \right\}  \\
&=&\E\left( Y|U,A,X\right) +\omega \left\{ W-\E\left( W|U,X\right) \right\}  \\
&=&b_{0}+b_{a}A+b_{x}X+b_{w}\E\left( W|U,X\right) +\omega \left\{ W-\E\left( W|U,X\right)
\right\}  \\
&=&b_{0}+b_{a}A+b_{x}X+\left( b_{w}-\omega \right) \E\left( W|U,X\right) +\omega W,
\end{eqnarray*}
where
\[
\E\left( W|U,X\right) =\E\left( W|U,A,Z,X\right) =\mu _{0}+\mu _{x}X+\frac{\sigma _{wu}}{
\sigma _{u}^{2}}\left( U-\kappa _{0}-\kappa_{x}X\right).
\]
The parameters are set as follows:
\begin{itemize}
\setlength\itemsep{1em}

\item $\Gamma_x=(0.25,0.25)^T$, $\Sigma_x=\left(
\begin{array}{ccc}
\sigma_x^2 & 0\\
0 & \sigma_x^2\\
\end{array}
\right)$, $\sigma_x=0.25$.

\item $\Pr \left(A=1|X\right)=\left[1+ \exp\{(0.125,0.125)^TX\}\right]^{-1}$.


\item $\alpha_0= 0.25$, $\alpha_a= 0.25$, $\alpha_x= (0.25,0.25)^T$.

\item $\mu_0= 0.25$, $\mu_a= 0.125$, $\mu_x= (0.25,0.25)^T$.

\item $\kappa_0= 0.25$, $\kappa_a= 0.25$, $\kappa_x= (0.25,0.25)^T$.

\item $\Sigma=\left(
\begin{array}{ccc}
1 & 0.25 & 0.5 \\
0.25 & 1 & 0.5 \\
0.5 & 0.5 & 1
\end{array}
\right), \sigma_y=0.25.$

\item $b_0= 2$, $b_a= 2$, $b_x= (0.25,0.25)^T$, $b_w=4$, $\omega=2$. 
\end{itemize}

As shown in the Appendix~\ref{sec:H}, the above data generating mechanism is compatible with the following models of $h$ and $q$ :
\begin{align}
h(W,A,X;b)&=b_{0}+b_{a}A+b_{w}W+b_{x}X, \label{eq1}\\
q(Z,A,X,t)&=1+\exp \left\{ \left( -1\right) ^{1-A}t_{0}+\left( -1\right)
^{1-A}t_{z}Z+\left( -1\right) ^{1-A}t_{a}A + \left( -1\right) ^{1-A}t_{x}X\right\},\label{eq2}
\end{align}
respectively, where $t_0=0.25, t_z=-0.5, t_a=-0.125,$ and $t_x= (0.25,0.25)^T$. All remaining parameter values are listed in the Appendix~\ref{sec:H}.

\subsection{Various estimators}

We implemented our three proposed proximal analogs of outcome regression (proximal OR), inverse probability weighted (proximal IPW), and doubly robust estimators  (proximal DR) of the causal effect. Confounding bridge functions $h$ and $q$ were estimated by solving estimating equations \eqref{est:1} and \eqref{est:2}, to yield $\widehat h$ and $\widehat q$, respectively under models ~\eqref{eq1} and \eqref{eq2}.
The resulting proximal OR, proximal IPW, and proximal DR estimators are denoted as $\widehat \psi_{POR}$, $\widehat \psi_{PIPW}$, $\widehat \psi_{PDR}$, respectively.
 Standard errors of estimators are computed using an empirical sandwich estimator obtained from standard theory of M-estimation \citep{Stefanski2002M}.
Furthermore, we compared proximal estimators to a standard doubly robust estimator, which is in principle, valid only under exchangeability.
The standard doubly robust estimator is given by
\begin{align*}
\widehat \psi_{DR} = \PP_n \left\{ \frac{(-1)^{1-A}}{\widehat f(A|L)} \{Y- \widehat \E[Y|L,A]\} + \widehat \E[Y|L,A=1] -\widehat \E[Y|L,A=0] \right\},
\end{align*}
where $\widehat f(A|L)$ and $\widehat \E[Y|L,A]$ are estimated via standard logistic regression and linear regression, respectively.

We consider four scenarios to investigate operational characteristics of our various estimators in a range of settings.
Following \cite{kang2007demystifying}, we evaluate the performance of the proposed estimators in situations where either or both confounding bridge functions are mis-specified by considering a model using a transformation of observed variables. 
In particular, covariates $X,W$ are used for $h$ and $X,Z$ are used for $q$ in the first scenario, i.e., the models are both correctly specified. 
In the second scenario, we use $W^*=|W|^{1/2}+3$ instead of $W$ for estimation of $h$, i.e., the outcome confounding bridge function is mis-specified. 
In the third scenario, we use  $Z^*=|Z|^{1/2}+3$ instead of $Z$ for estimation of $q$, i.e., the treatment confounding bridge function is mis-specified.
In the fourth scenario, we consider the case in which both confounding bridge functions are mis-specified, i.e., we use 
$W^* = |W|^{1/2}+1$ and $Z^* = |Z|^{1/2}+1$ to estimate $h$ and $q$, respectively. 
Moreover, for each scenario, transformed variables are used to estimate outcome regression and treatment process respectively in standard doubly robust estimator.
We consider sample size $n=2000$ and each simulation is repeated \yifan{500} times.

\subsection{Numerical results}

\begin{table}[h]
\begin{center}
\caption{\label{simu1}
Simulation results: absolute bias ($\times 10^{-2}$) and MSE ($\times 10^{-2}$)}
\begin{tabular}{cccccc}
\noalign{\smallskip}
\noalign{\smallskip}
 & & $\widehat \psi_{DR}$ & $\widehat \psi_{POR}$ & $\widehat \psi_{PIPW}$ & $\widehat \psi_{PDR}$ \\
\noalign{\smallskip}
\multirow{2}{*}{Scenario~1}  & Bias &  9.9  &  0.4  &  0.5  &  0.5 \\
                          & MSE &  1.1  &   0.6  &  0.7 & 0.7   \\
\noalign{\smallskip}
\multirow{2}{*}{Scenario~2} & Bias  &  23.4 &  64.0  &   0.5  & 0.1 \\
                & MSE &  6.5   &  47.5  & 0.7 & 3.0 \\
\noalign{\smallskip}
\multirow{2}{*}{Scenario~3} & Bias  & 17.3 &  0.4  &  20.4  & 0.2  \\
                & MSE  &  3.2  &  0.6  &  4.5 & 0.7  \\
                \noalign{\smallskip}
\multirow{2}{*}{Scenario~4} & Bias  & 43.2 &  30.0  &  13.7  & 34.5  \\
                & MSE  &  19.9  &  19.2  &  2.3 & 23.5 \\
\noalign{\smallskip}
\end{tabular}
\end{center}
\source{\small{$\widehat \psi_{DR}$ refers to the standard doubly robust estimator; $\widehat \psi_{POR}$ refers to the proximal outcome regression approach; $\widehat \psi_{PIPW}$ refers to the proximal treatment confounding bridge approach; $\widehat \psi_{PDR}$ refers to the proximal doubly robust estimator.}}
\end{table}

\begin{table}[h]
\begin{center}
\caption{\label{simu2}
Simulation results: coverage ($\%$) and average length ($\times 10^{-2}$)}
\begin{tabular}{cccccc}
\noalign{\smallskip}
\noalign{\smallskip}
 & & $\widehat \psi_{DR}$ & $\widehat \psi_{POR}$ & $\widehat \psi_{PIPW}$ & $\widehat \psi_{PDR}$ \\
\noalign{\smallskip}
\multirow{2}{*}{Scenario~1}  & Coverage &  25.4   &  95.2  &  95.8  &  95.8  \\
                          & Length &  14.4  &  31.9  &  33.1 & 33.1   \\
\noalign{\smallskip}
\multirow{2}{*}{Scenario~2} & Coverage  &  37.6 &  14.0  &   95.8  & 98.6  \\
                & Length &  41.2  &  71.2  &  33.1 & 69.5  \\
\noalign{\smallskip}
\multirow{2}{*}{Scenario~3} & Coverage  & 0.2 &  95.2  &   26.0  & 96.6  \\
                & Length  &  15.6  &  31.9  &  33.6 &  34.2  \\
                \noalign{\smallskip}
\multirow{2}{*}{Scenario~4} & Coverage  & 3.0 &  83.0  &  72.2  & 90.8  \\
                & Length  &  43.8  &  130.3  &  37.9 & 152.1 \\
\noalign{\smallskip}
\end{tabular}
\end{center}
\source{}
\end{table}

Table~\ref{simu1} summarizes simulation results. As expected, the proximal IPW estimator has small bias in Scenarios~1 and 2, and the proximal OR estimator has small bias in Scenarios~1 and 3. 
The proximal doubly robust estimator has small bias in the first three scenarios. 
In Scenario~4, the doubly robust proximal estimator has similar bias to both proximal IPW and proximal OR estimators due to model misspecification. 
In addition, as expected, the standard doubly robust estimator is severely biased in all four scenarios due to unmeasured confounding.

Table~\ref{simu2} presents coverage and average length of 95\% confidence intervals. In agreement with semiparametric theory, proximal OR yields the narrowest confidence intervals with nominal coverage when $h$ is correctly specified, i.e., Scenarios~1 and 3. Likewise proximal IPW confidence intervals have correct coverage in Scenarios~1 and 2, while proximal doubly robust approach has nominal coverage in all first three scenarios.

In the Appendix~\ref{sec:addi}, we also report additional simulation results for the case where $U$ is not a confounder, and sensitivity analysis on violation of Assumptions~4 and 5, and on the strength of dependence between $Z$ and $W$ given $X$ and $A$.

\section{Data analysis}\label{sec:real}

	In this section, we illustrate the proposed semiparametric proximal estimators of the ATE in a data application considered in \cite{tchetgen2020}. The Study to Understand Prognoses
and Preferences for Outcomes and Risks of Treatments (SUPPORT) to evaluate the effectiveness of right heart catheterization (RHC) in the intensive care unit of critically ill patients \citep{5c6af36c0fb64cfcbb482d75c2bc7ff1}. These data have been re-analyzed in a number of papers in causal inference literature under a key exchangeability condition on basis of measured covariates; including  \cite{tan2006distributional,2015biasreduce,tan2019model,tan2019regularized,cui2019selective}. 

	We consider the effect of RHC on 30-day survival. Data are available on 5735 individuals, 2184 treated and 3551 controls. In total, 3817 patients survived and 1918 died within 30 days. The outcome $Y$ is the number of days between admission and death or censoring at day 30. Similar to \cite{Hirano2001} and \cite{tchetgen2020}, we include 71 baseline covariates to adjust for potential confounding,
including demographics (such as age, sex, race, education, income, and insurance status), estimated
probability of survival, comorbidity, vital signs, physiological status, and functional status. See Table 1 of \cite{Hirano2001} for further details. 

In order to address a general concern that patients either self-selected or were selected by their physician to take the treatment based on information not recorded in the data set; we performed an analysis using the proximal framework and methods developed in this paper.
Ten variables measuring the patients' overall physiological status were measured from a blood test during the initial 24 hours in the intensive care unit. These variables might be subject to substantial measurement error
and as a single snapshot of the underlying physiological state over time may be viewed as potential confounding proxies.
Among those ten physiological status measures, pafi1, paco21, ph1, and
hema1 are strongly correlated with both the treatment and the outcome.
As in \cite{tchetgen2020},  we allocated $Z$ = (pafi1, paco21) and $W$ = (ph1, hema1), and collected the remaining 67 variables in $X$. 
\yifan{The tests for pairwise partial correlations \citep{kim2015ppcor} between $Z$ and $W$ given $X$ and $A$ are all significant at the 0.05 level.}
We specified the outcome counfounding bridge function according to Equation~\eqref{eq1}, and specified the model given by  Equation~\eqref{eq2} for the treatment confounding bridge function, including interaction terms $Z$-$A$ and $X$-$A$ to improve goodness of fit.
 
Table~\ref{real} presents point estimates and corresponding 95\% confidence intervals for the average treatment effect.
 The proximal doubly robust estimator, proximal IPW estimator, and proximal outcome regression estimator are much larger than the standard doubly robust estimator.
Concordance between the three proximal estimators offers confidence in modeling assumptions, indicating that RHC may have a more harmful effect on 30 day-survival among critically ill patients admitted into an intensive care unit than previously reported.

 \yf{We highlight that the proximal causal inference framework is an alternative to traditional methods: instead of requiring no unmeasured confounding, we require validity of the treatment- and outcome-inducing confounding proxies.
If the choice of $Z$ and $W$ does not meet Assumptions~\ref{asm:condindy} or \ref{asm:condindw}, the proximal causal estimators can be biased. This is a potential limitation of the proximal causal inference framework. 
 Therefore, the selection of valid treatment and outcome proxies should be based on reliable subject matter knowledge because Assumptions~\ref{asm:condindy}, \ref{asm:condindw}, and completeness conditions must be met.} In the Appendix~\ref{sec:addi}, we perform a sensitivity analysis in which a variable is removed from $Z$ (pafi1 or paco21) or $W$ (ph1 or hema1). The results suggest that the proxies may not be equally relevant to the potential source of unmeasured confounding, although the totality of evidence is well aligned with the results given in Table~\ref{real}.

\begin{table}[h]
\begin{center}
\caption{\label{real}
Treatment effect estimates (standard deviations) and 95\% confidence intervals of the average treatment effect}
\begin{tabular}{ccccc}
\noalign{\smallskip}
\noalign{\smallskip}
 & $\widehat \psi_{DR}$ & $\widehat \psi_{POR}$ & $\widehat \psi_{PIPW}$ & $\widehat \psi_{PDR}$ \\
\noalign{\smallskip}
Treatment effects (SDs) &  -1.17 (0.32)   &  -1.80 (0.44)  &  -1.72 (0.30)  &   -1.66 (0.43)\\
 95\% CIs &  (-1.79,-0.55)  &   (-2.65,-0.94)  &  (-2.30,-1.14) & (-2.50,-0.83)   \\
\noalign{\smallskip}
\end{tabular}
\end{center}
\end{table}

\section{Discussion}\label{sec:discussion}

In this paper, we have provided a new condition for nonparametric proximal identification of the population average treatment effect. We have also derived the semiparametric locally efficiency bound for estimating the corresponding identifying functional under three key semiparametric models: (i) one in which the observed data law is solely restricted by a parametric model for the outcome confounding bridge function; (ii) one in which the observed data law is solely restricted by a parametric model for the treatment confounding bridge function; and (iii) a model in which both confounding bridge functions are unrestricted. For inference, we have provided a large class of doubly robust locally efficient estimators for the ATE, which attain the efficiency bound for the model given in (iii) at the intersection submodel where both (i) and (ii) hold. In addition, we have also derived analogous results for the average treatment effect on the treated.
Our approach was illustrated via
 simulation studies and a real data application. 
Our paper contributes to the literature on doubly robust functionals as our semiparametric proximal causal inference approach can be viewed as a generalization of standard doubly robust estimators \citep{robinsetal1994} to account for potential unmeasured confounding by leveraging a pair of proxy variables.

The proposed methods may be improved or extended in several directions.
\yf{
Our proximal causal inference framework relies on the validity of treatment- and outcome-inducing confounding proxies. 
When Assumption~\ref{asm:condindy} or \ref{asm:condindw} is violated,  the proximal causal inference estimators can be biased even if exchangeability on the basis of measured covariates holds. Therefore, future research should study methods to carefully sort out proxies when domain knowledge is lacking.
} 
\yf{Moreover, because using semiparametric models as an implicit form of regularization involves modeling assumptions on bridge functions which the data analyst might sometimes have little a priori understanding, we caution that using the semiparametric model restrictions to identify bridge functions might be disadvantageous in the absence of subject-matter knowledge \citep{rotnitzky1997}.} 
\yf{In addition, given that models for the bridge functions might sometimes be difficult to postulate, it is important to further develop model checking tools for bridge functions in future research.} 

Another important direction is to characterize the minimax rate of estimation of the average treatment effect in settings where the smoothness of nuisance parameters is very low and/or the dimension of proxies and other observed covariates is high, so that root-$n$ estimation rates of the causal effect functional may no longer be possible. Upon establishing new lower bound rates, we plan to construct estimators to achieve minimax optimal rates by leveraging the efficient influence function obtained in the current paper as well as higher order influence functions needed to further correct higher order biases thus effectively extending the work of Robins and colleagues \citep{robins2008HOIF,robins2017minimax}. Other directions for future research include proxy selection and validation methods, as well as partial identification results in case completeness conditions needed for point identification are only approximately true.

\newpage
\appendix
\begin{center}
\LARGE Appendix
\end{center}

\section{Proof of Theorem~\ref{thm:identification}}

\begin{proof}
(1) We first show that
\begin{align*}
\frac{1}{f(A=a|U=u,X=x)} = \int q(z,a,x)dF(z|U=u,A=a,X=x).
\end{align*}
To show this, note that $W \perp (Z, A) | U, X$ and
\begin{align*}
\E[q(Z,a,x)|W=w,A=a, X=x] = \frac{1}{f(A=a|W=w, X=x)},
\end{align*}
so we have that
\begin{align*}
 & \int \frac{1}{f(A=a|u,x)} dF(u|w,a,x)\\
= & \int \frac{1}{f(A=a|u,w,x)} dF(u|w,a,x)\\
=& \frac{1}{f(A=a|w,x)}\\
= & \E[q(Z,a,x)|W=w,A=a, X=x]\\
=& \int \int q(z,a,x)dF(z|u,w,a,x)dF(u|w,a,x)\\
= & \int \int q(z,a,x)dF(z|u,a,x)dF(u|w,a,x).
\end{align*}

Furthermore, by the completeness assumption~\ref{asm:completeness2}(1), we have
\begin{align*}
\frac{1}{f(A=a|U,x)} = \int q(z,a,x)dF(z|U,a,x),
\end{align*}
almost surely. 
Thus, by $Y \perp Z | (U, A, X)$,
\begin{align*}
 &\E[I(A=a)Yq(Z,a,X)|U,X=x] ~~\yf{(\int_{\mathcal Y,\mathcal Z} yq(z,a,x)dF(y,z,a|U,x) \text{~if~} A \text{~is~continuous})}\\
= &\E[Yq(Z,a,X)|U,A=a,X=x]f(A=a|U,X=x)\\
= &\E[Y|U,A=a,X=x]\E[q(Z,a,X)|U,A=a,X=x]f(A=a|U,X=x)\\
= & \E[Y|U,A=a,X=x].
\end{align*}

(2) By Assumptions~\ref{asm:consistency} and \ref{asm:condrand},
\begin{align*}
\E[Y|U,A=a,X=x] = &\E[Y(a)|U,X=x].
\end{align*}

Therefore, given that Equation~\eqref{eq:u3} holds,
\begin{align*}
& \E[Y(a)]\\
= &\int_\cX \E[Y(a)|X=x]dF(x)\\
= & \int_{\cX} \int I(\tilde a=a)q(z,a,x)y dF(y,z,\tilde a|x)dF(x) ~~\yf{(\int_{\mathcal X}\int_{\mathcal Y,\mathcal Z} yq(z,a,x)dF(y,z,a|x)dF(x) \text{~if~} A \text{~is~continuous})},
\end{align*}
which completes the proof.
\end{proof}

\section{Existence of solutions to Equations~\eqref{eq:condexpec} and \eqref{eq:condexpect2}}\label{sec:existence}

We consider the singular value decomposition (\cite{carrasco2007linear}, Theorem 2.41) of compact operators to characterize conditions for existence of a solution to Equations~\eqref{eq:condexpec} and \eqref{eq:condexpect2}.
The former one is the same as that of \cite{miao2018identifying}. 
Let $L^2\{F(t)\}$ denote the space of all square integrable functions of $t$ with respect to a cumulative distribution function
$F(t)$, which is a Hilbert space with inner product $\langle g_1, g_2 \rangle=\int g_1(t)g_2(t)dF(t)$.
Similar to Assumption~\ref{asm:bound}, let $T_{a,x}$ denote the operator: $L^2\{F(w|a,x)\}\rightarrow L^2\{F(z|a,x)\}$, $T_{a,x}h=\E[h(W)|z,a,x]$ and let $(\lambda_{a,x,n},\varphi_{a,x,n},\phi_{a,x,n})_{n=1}^{\infty}$ denote
a singular value decomposition of $T_{a,x}$.
Also let $T'_{a,x}$ denote the operator: $L^2\{F(z|a,x)\}\rightarrow L^2\{F(w|a,x)\}$, $T'_{a,x}q=\E[q(Z)|w,a,x]$ and let $(\lambda'_{a,x,n},\varphi'_{a,x,n},\phi'_{a,x,n})_{n=1}^{\infty}$ denote
a singular value decomposition of $T'_{a,x}$.
Furthermore, we assume the following regularity conditions:\\
1)$\int\int f(w|z,a,x)f(z|w,a,x) dwdz < \infty$;\\
2)$\int \E^2[Y|z,a,x] f(z|a,x) dz < \infty$;~~
2')$\int f^{-2}(a|w,x)f(w|a,x) dw < \infty$;\\
3)$\sum_{n=1}^{\infty} \lambda_{a,x,n}^{-2}|\langle \E[Y|z,a,x],\phi_{a,x,n}\rangle |^2<\infty$;\\
3')$\sum_{n=1}^{\infty} \lambda_{a,x,n}^{'-2}|\langle f^{-1}(a|w,x),\phi'_{a,x,n}\rangle |^2<\infty$.

Given $\E(Y|z,a,x)$, $f(w|z,a,x)$, or $f(a|w,x)$, $f(z|w,a,x)$, the solution to \eqref{eq:condexpec} or \eqref{eq:condexpect2} exists if Assumption~\ref{asm:completeness1.2} and conditions (1)(2)(3) or Assumption~\ref{asm:completeness2.2} and conditions (1)(2')(3') hold, respectively.
The proof follows immediately from Picard's theorem \citep{kress1989linear}.

\section{Theorem~\ref{ap:thm} and its proof}\label{sec:C}

\begin{theorem}\label{ap:thm}
\yifan{Suppose Assumptions~\ref{asm:condindw}, \ref{asm:positivity2} hold and that there exists a bridge function $q'(z, a, x)$} such that
${1}/{f(A=a|U,X)} = \int q'(z,a,X)dF(z|U,A=a,X)$ almost surely. Furthermore, suppose that for any square-integrable function $g$ and for any $a,x$, $\E [g(Z)|W, A=a, X=x] = 0$ almost surely if and only if $g(Z) = 0$ almost surely, we have that
Equation~\eqref{eq:condexpect2} has a unique solution $q'$.
\end{theorem}
\begin{proof}
We need to show that
\begin{align*}
\frac{1}{f(A=a|W=w,X=x)} = \int q'(z,a,x)dF(z|W=w,A=a,X=x).
\end{align*}
To show this, note that $W \perp (Z, A) | U, X$ and
\begin{align*}
\E[q'(Z,a,x)|U=u, A=a, X=x] = \frac{1}{f(A=a|U=u, X=x)},
\end{align*}
so we have that
\begin{align*}
&  \frac{1}{f(A=a|w,x)} \\
= & \int \frac{1}{f(A=a|u,w,x)} dF(u|w,a,x)\\
= & \int \frac{1}{f(A=a|u,x)} dF(u|w,a,x)\\
= &  \int \E[q'(Z,a,x)|U=u, A=a, X=x] dF(u|w,a,x)\\
=& \int \int q'(z,a,x)dF(z|u,w,a,x)dF(u|w,a,x)\\
=& \int q'(z,a,x)dF(z|w,a,x).
\end{align*}
Therefore, 
\begin{align*}
\frac{1}{f(A=a|w,x)} = \int q'(z,a,x)dF(z|w,a,x).
\end{align*}
For the uniqueness, suppose both $q'_1(z,a,x)$ and $q'_2(z,a,x)$ satisfy the above equation, then one must have that for all $w$, $a$, and $x$,
$$\int [q'_1(z,a,x)-q'_2(z,a,x)]dF(z|w,a,x)=0.$$
By the completeness condition, $q'_1(z,a,x)$ must equal $q'_2(z,a,x)$ almost surely.
Thus, Equation~\eqref{eq:condexpect2} has a unique solution $q'$.
\end{proof}

\section{Proof of Theorem~\ref{thm:eif}}
\begin{proof}
(1)
For the uniqueness of $h$, suppose both $h_1(w,a,x)$ and $h_2(w,a,x)$ satisfy Equation~\eqref{eq:condexpec}, then one must have that for all $z$, $a$, and $x$,
$$\int [h_1(z,a,x)-h_2(z,a,x)]dF(w|z,a,x)=0.$$
By the completeness assumption~\ref{asm:completeness2.2}, $h_1(w,a,x)$ must equal $h_2(w,a,x)$ almost surely.
Thus, the solution $h$ is unique.

For the uniqueness of $q$, suppose both $q_1(z,a,x)$ and $q_2(z,a,x)$ satisfy Equation~\eqref{eq:condexpect2}, then one must have that for all $w$, $a$, and $x$,
$$\int [q_1(z,a,x)-q_2(z,a,x)]dF(z|w,a,x)=0.$$
By the completeness assumption~\ref{asm:completeness1.2}, $q_1(z,a,x)$ must equal $q_2(z,a,x)$ almost surely.
Thus, the solution $q$ is unique.

(2)
In order to find the efficient influence function for $\psi$, we need to first find a random variable $G$ with mean 0 and
\begin{align}\label{eq:score}
\frac{\partial \psi_t}{\partial t}|_{t=0} = \E[GS(\cO;t)]|_{t=0},
\end{align}
where $S(\cO;t)=\partial \log f(\cO;t)/\partial t$, and $\psi_t$ is the parameter of interest $\psi$ under a regular parametric submodel in $\mM_{sp}$ indexed by $t$ that includes the true data generating mechanism at $t = 0$ \citep{vaart_1998}.

Recall that
\begin{align*}
\E[Y|Z,A,X] = \int h(w,A,X)dF(w|Z,A,X),
\end{align*}
so we have
\begin{align*}
\partial \E_t[Y-h_t(W,A,X)|Z,A,X]/\partial t|_{t=0}=0.
\end{align*}
Thus,
\begin{align*}
\int \frac{\partial [\{y-h_t(w,A,X)\}f_t(w,y|Z,A,X)]}{\partial t}|_{t=0} d(w,y) = 0.
\end{align*}
Let $\epsilon = Y - h(W,A,X)$, we have that
\begin{align*}
\E[\epsilon S(W,Y|Z,A,X)|Z,A,X] = \E[\partial h_t(W,A,X)/\partial t|_{t=0}|Z,A,X].
\end{align*}

The left hand side of Equation~\eqref{eq:score} equals to
\begin{align*}
\partial \psi_t/\partial t|_{t=0} =& \partial [\int \int \{h_t(w,1,x)-h_t(w,0,x)\} dF_t(w|x)dF_t(x)]/\partial t|_{t=0} \\
=& \E[\{h(W,1,X)-h(W,0,X)\}S(W,X)]\\& + \left[\int \int \partial h_t(w,1,x)/\partial t|_{t=0} dF(w|x)dF(x) - \int \int \partial h_t(w,0,x)/\partial t|_{t=0} dF(w|x)dF(x)\right].
\end{align*}
The first term is equal to
\begin{align*}
& \E[(h(W,1,X)-h(W,0,X)-\psi)(S(Z,Y,A|W,X)+S(W,X))]\\
= & \E[(h(W,1,X)-h(W,0,X)-\psi)S(\cO)].
\end{align*}
The second term is equal to
\begin{align*}
& \int \int \partial h_t(w,1,x)/\partial t|_{t=0} dF(w|x)dF(x) - \int \int \partial h_t(w,0,x)/\partial t|_{t=0} dF(w|x)dF(x)\\
= &\E [\frac{(-1)^{1-A}}{f(A|W,X)}\partial h_t(W,A,X)/\partial t|_{t=0} ] \\
= & \E[(-1)^{1-A}q(Z,A,X)\partial h_t(W,A,X)/\partial t|_{t=0} ]\\
= & \E[(-1)^{1-A}q(Z,A,X)\E[\partial h_t(W,A,X)/\partial t|_{t=0}|Z,A,X] ]\\
= & \E[(-1)^{1-A}q(Z,A,X)\epsilon S(W,Y|Z,A,X) ]\\
= & \E[(-1)^{1-A}q(Z,A,X)\epsilon S(W,Y|Z,A,X) ] + \E[(-1)^{1-A}q(Z,A,X)\epsilon S(Z,A,X) ]\\
= & \E[(-1)^{1-A}q(Z,A,X)\epsilon S(\cO) ].
\end{align*}
Finally, combing two terms gives
\begin{align*}
\frac{\partial \psi_t}{\partial t}|_{t=0} =\E\left[ \left\{ (-1)^{1-A} q(Z,A,X) [Y- h(W,A,X)] + h(W,1,X)-h(W,0,X)- \psi  \right\} S(\cO)\right].
\end{align*}
Therefore, \begin{align}\label{eq:if}
(-1)^{1-A} q(Z,A,X) [Y- h(W,A,X)] + h(W,1,X)-h(W,0,X)- \psi,
\end{align}
is an influence function of $\psi$. 
Next, we show that the influence function~\eqref{eq:if}
belongs to the tangent space 
\begin{align*}
&   \Lambda_1 + \Lambda_2 \\ \equiv & \{S(Z,A,X) \in L_2(Z,A,X):\E[S(Z,A,X)]=0\}\\& + \{S(Y,W|Z,A,X)\in L_2(Z,A,X)^\perp:\E[\epsilon S(Y,W|Z,A,X)|Z,A,X]\in cl(R(T))\},
\end{align*}
where $R(T)$ denotes the range space of $T$, $A^\perp$ denotes the orthogonal complement of $A$, and $cl(A)$ refers to the closure of $A$.
To see this, 
note that we have the following decomposition of Equation~\eqref{eq:if}, 
\begin{align*}
& (-1)^{1-A} q(Z,A,X) [Y- h(W,A,X)] + h(W,1,X)-h(W,0,X)- \psi\\
= & \{\E[h(W,1,X)-h(W,0,X)- \psi|Z,A,X]\}\\&  + h(W,1,X)-h(W,0,X)-\psi-\{\E[h(W,1,X)-h(W,0,X)- \psi|Z,A,X]\}\\& + (-1)^{1-A} q(Z,A,X) [Y- h(W,A,X)].
\end{align*}
We have
$ \{\E[h(W,1,X)-h(W,0,X)- \psi|Z,A,X]\}\in \Lambda_1$, and the remaining part 
belongs to $\Lambda_2$ as
\begin{align*}
& \E[\epsilon\{h(W,1,X)-h(W,0,X)-\psi\}-\epsilon\{\E[h(W,1,X)-h(W,0,X)- \psi|Z,A,X]\}|Z,A,X]
\in cl(R(T)),\\
 &  \E[ (-1)^{1-A} q(Z,A,X)\epsilon^2|Z,A,X] \in cl(R(T)),
\end{align*}
because of Assumption~\ref{asm:bound}, and
\begin{align*}
&\E[\{h(W,1,X)-h(W,0,X)-\psi\}-\{\E[h(W,1,X)-h(W,0,X)- \psi|Z,A,X]\}|Z,A,X]=0,\\
&\E[ (-1)^{1-A} q(Z,A,X)\epsilon|Z,A,X]=0,
\end{align*}
Therefore it completes the proof.

3) The proof of double robustness we refer to Section~\ref{sec:F}.

\end{proof}

\section{Theorems~\ref{thm:if} and \ref{thm:eifappendix} and their proofs}\label{sec:E}

\begin{theorem}\label{thm:if} 
a) Let $h(W,A,X;b)$ denote the parametric component of the semiparametric model $\mathcal M_1$ with unknown finite dimensional parameter $b$.
Then, all influence functions of regular and asymptotically linear (RAL) estimators of $b$ in  $\mathcal M_1$ are of the form: 
\begin{align}\label{eq:inff1}
\left\{\E\left[\frac{\partial h(W,A,X;b)m(Z,A,X)}{\partial b}\right]\right\}^{-1} \{Y -  h(W, A, X;b) \} m(Z,A,X),
\end{align}
for some function $m(Z,A,X)$ of the same dimension as $b$; furthermore, the semiparametric efficient influence function of $b$ in model $\mathcal M_1$ is given by $m_{eff}(Z,A,X)$ provided in Theorem~\ref{thm:eifappendix}.  

b) Let $q(W,A,X;t)$ denote the parametric component of the semiparametric model $\mathcal M_2$ with unknown finite dimensional parameter $t$. Then, all influence functions of regular and asymptotically linear (RAL) estimators of $t$ in $\mathcal M_2$ are of the form: 
\begin{align}\label{eq:inff2}
-\left\{\E\left[\frac{\partial q(Z,A,X;t) n(W,A,X)}{\partial t}\right]\right\}^{-1}[q(Z,A,X;t) n(W,A,X) -n(W,1,X) -n(W,0,X)],
\end{align}
where $n(W,A,X)$ is any function having the same dimension as $t$; furthermore, the semiparametric efficient influence function of $t$ in model $\mathcal M_2$ is given by $n_{eff}(W,A,X)$ provided in Theorem~\ref{thm:eifappendix}.  

\end{theorem}

\begin{theorem}\label{thm:eifappendix}We consider the semiparametric models of Theorem~\ref{thm:if}:

a) The semiparametric efficient influence function of $b$ in model $\mathcal M_1$ is given by \eqref{eq:inff1} with
\begin{align*}
m_{eff}\left( Z,A,X\right) = \E\left\{ \left[ Y-h\left( W,A,X;b\right) \right]
^{2}|Z,A,X\right\} ^{-1}\E\left\{ \nabla _{b}h\left( W,A,X;b\right)
|Z,A,X\right\}.
\end{align*}

b) The semiparametric efficient influence function of $t$ in model $\mathcal M_2$ is given by
\eqref{eq:inff2} with
\[
n_{eff}\left( W,A,X\right) \ =\left\{ 
\begin{array}{c}
\E\left\{ \left. S\left( A|W,X\right) \times \left\{ \frac{A}{f\left(
A|W,X\right) }-1\right\} \right\vert W,X\right\}\frac{\Omega(W,A,X)}{
\E\left\{ \left. \left\{ q\left( Z,A,X\right) -\frac{1}{f\left( A|W,X\right) }
\right\} ^{2}\right\vert W,A,X\right\} }  \\ 
-\left[ 
\begin{array}{c}
v\left(W,X\right) 
\times \E\left\{ \left. \left\{ \frac{A}{f\left( A|W,X\right) }-1\right\}
^{2}\right\vert W,X\right\} \frac{\Omega(W,A,X)}{
\E\left\{ \left. \left\{ q\left( Z,A,X\right) -\frac{1}{f\left( A|W,X\right) }
\right\} ^{2}\right\vert W,A,X\right\} }
\end{array}
\right] 
\end{array}
\right\},
\]
where
\begin{align*}
\Omega(W,A,X)\equiv \left[ A\frac{1\ }{f\left( 1|W,X\right) }%
-\left( 1-A\right) \frac{f\left( 1|W,X\right) }{\left[ 1-f\left(
1|W,X\right) \right]^{2}}\right].
\end{align*}
\end{theorem}

In the following, we provide the proofs of Theorems~\ref{thm:if} and \ref{thm:eifappendix}.
\begin{proof}
a) Recall that
\begin{align*}
\E[Y -  h(W, A, X;b)  |Z,A,X] = 0,
\end{align*}
so we have 
\begin{align*}
\partial \E_s[ \{Y -  h(W, A, X;b_s) \} m(Z,A,X)] /\partial s|_{s=0} = 0,
\end{align*}
for any function $m(Z,A,X)$ having the same dimension as $b$.
Therefore,
\begin{align*}
\E\left[\frac{\partial h(W,A,X;b)m(Z,A,X)}{\partial b}\right]\frac{\partial b_s}{\partial s}|_{s=0} =
\E[ \{Y -  h(W, A, X;b) \} m(Z,A,X)S(\cO)].
\end{align*}

To show the efficient influence function for outcome confounding bridge function, consider the set of influence functions (up to constant matrix multipliers),
\begin{align*}
\left[ Y-h\left( W,A,X;b\right) \right] m\left( Z,A,X\right).
\end{align*}
Then using a result due to \cite{NEWEY19942111}, one can show that 
\begin{align*}
m_{eff}\left( Z,A,X\right) = \E\left\{ \left[ Y-h\left( W,A,X;b\right) \right]
^{2}|Z,A,X\right\} ^{-1}\E\left\{ \nabla _{b}h\left( W,A,X;b\right)
|Z,A,X\right\}.
\end{align*}

b) Recall that
\begin{align*}
\E \left[q(Z,A,X;t) - \frac{1}{f(A|W,X)}\bigg|W,A,X\right] = 0,
\end{align*}
so we have 
\begin{align*}
\partial \E_s \left[\bigg\{q(Z,A,X;t_s) - \frac{1}{f_s(A|W,X)}\bigg\}n(W,A,X)\right]/\partial s\bigg|_{s=0}= 0,
\end{align*}
for any function $n(W,A,X)$ having the same dimension as $t$.

Note that
\begin{align*}
    \frac{\partial}{\partial s}\frac{1}{f_s(A|W,X)} = - \frac{\partial f_s(A|W,X)}{f^2(A|W,X)},
\end{align*}
so we have that
\begin{align*}
&-\E\left[\frac{\partial q(Z,A,X;t) n(W,A,X)}{\partial t}\right]\frac{\partial t_s}{\partial s}|_{s=0}\\&= \E\left[\left\{ q(Z,A,X;t) - \frac{1}{f(A|W,X)} \right\}n(W,A,X)S(\cO)\right]\\& + \E\left[\left\{  \frac{n(W,A,X)}{f(A|W,X)} -  n(W,1,X) - n(W,0,X)\right\}S(\cO)\right]\\
& = \E\left[ \left\{ q(Z,A,X;t) n(W,A,X) -n(W,1,X) -n(W,0,X) \right\} S(\cO)\right].
\end{align*}

To show the efficient influence function for treatment confounding bridge function, consider the semiparametric model with sole restriction that $q\left(
Z,A,X;t\right) $ belongs to a parametric model with unknown parameter such that 
\begin{align*}
\frac{1}{f\left( A|W,X;t\right) }=\int q\left( z,A,X;t\right) dF\left(
z|W,A,X\right). 
\end{align*}
The score equation of $t$ is therefore 
\begin{align*}
S\left( A|W,X\right)  =&\frac{\nabla _{t}f\left( A|W,X;t\right) }{f\left(
A|W,X\right) } \\
=&-f\left( A|W,X\right) \int \nabla _{t}q\left( z,A,X;t\right) dF\left(
z|W,A,X\right)  \\
=&-\frac{\int \nabla _{t}q\left( z,A,X;t\right) dF\left( z|W,A,X\right) }{
\int q\left(z,A,X\right) dF\left( z|W,A,X\right) }.
\end{align*}
The set of influence functions of $t$ evaluated at the true parameter includes functions of form (up to constant matrix multipliers):
\begin{align*}
IF\left( n\right)  =& q\left( Z,A,X\right)
n\left( W,A,X\right) -n\left( W,1,X\right) - n\left( W,0,X\right). 
\end{align*}
We therefore have that 
\begin{align*}
&IF\left( n\right) -\E\left[ IF\left( n\right) |W,A,X\right]
+\E\left[ IF\left( n\right) |W,A,X\right]  \\
= & \left\{q\left( Z,A,X\right)- \frac{1}{f(A|W,X)}\right\}
n\left( W,A,X\right) \\
&+\left\{ \frac{A\ }{f\left( A|W,X\right) }-1\right\} n\left(W,1,X\right)  +\left\{\frac{1-A}{f\left( A|W,X\right) } -1\right\} n\left(W,0,X\right) \\
= & \left\{q\left( Z,A,X\right)- \frac{1}{f(A|W,X)}\right\}
n\left( W,A,X\right) \\
& +\left\{ \frac{A\ }{f\left( A|W,X\right) }-1\right\}
\left[n(W,1,X) - \frac{f(1|W,X)}{1-f(1|W,X)} n(W,0,X)\right]\\
= & \left\{q\left( Z,A,X\right)- \frac{1}{f(A|W,X)}\right\}
n\left( W,A,X\right) \\
& +\left\{ \frac{A\ }{f\left( A|W,X\right) }-1\right\}
\E\left[ n(W,A,X)  \left\{\frac{A}{f(1|W,X)} - \frac{[1-A]f(1|W,X)}{[1-f(1|W,X)]^2}\right\} \bigg |W,X\right].
\end{align*}

The first term is orthogonal to $S(A|W,X)$ as it has mean
zero given $\left(W,A,X\right)$. The efficient influence function $IF(n_{eff})$ is given by the solution to: for all $n(W,A,X)$,
\begin{eqnarray*}
&&0\\
&=&\E\left\{ \left[ S\left( A|W,X\right) -IF\left( n_{eff}\right) 
\right] IF\left( n\right) \right\}  \\
&=&\E\left\{ 
\begin{array}{c}
\left[ 
\begin{array}{c}
S\left( A|W,X\right) -\left\{ q\left( Z,A,X\right) -\frac{1}{f\left(
A|W,X\right) }\right\} n_{eff}\left( W,A,X\right)  \\ 
-\left\{ \frac{A}{f\left( A|W,X\right) }-1\right\} \times \left\{ \E\left(
\left. n_{eff}\left(W,A,X\right) \Omega\left(W,A,X\right) \right\vert W,X\right) \right\} 
\end{array}
\right]  \\ 
\left\{ 
\begin{array}{c}
\left\{ q\left( Z,A,X\right) -\frac{1}{f\left( A|W,X\right) }\right\} n\left(W,A,X\right)  \\ 
+\left\{ \frac{A}{f\left( A|W,X\right) }-1\right\} \times \left\{ \E\left(
\left. n\left( W,A,X\right) \Omega\left(W,A,X\right) \right\vert W,X\right) \right\} 
\end{array}
\right\} 
\end{array}
\right\}  \\
&=&\E\left\{ \left[ 
\begin{array}{c}
S\left( A|W,X\right) \times \left\{ \frac{A}{f\left( A|W,X\right) }
-1\right\} \times \left\{ \E\left( \left. n\left(W,A,X\right) \Omega\left(W,A,X\right)
\right\vert W,X\right) \right\}  \\ 
-\left\{ q\left( Z,A,X\right) -\frac{1}{f\left( A|W,X\right) }\right\}
n_{eff}\left( W,A,X\right) \left\{ q\left(
Z,A,X\right) -\frac{1}{f\left( A|W,X\right) }\right\} n\left( W,A,X\right)  \\ 
-\left[ 
\begin{array}{c}
 \left\{ \frac{A}{f\left( A|W,X\right) }-1\right\} \times \left\{ \E\left(
\left. n_{eff}\left( W,A,X\right) \Omega\left(W,A,X\right) \right\vert W,X\right) \right\}  \\ 
\times \left\{ \frac{A}{f\left( A|W,X\right) }-1\right\} \times \left\{ \E\left(
\left. n\left( W,A,X\right) \Omega\left(W,A,X\right) \right\vert W,X\right) \right\} 
\end{array}
\right] 
\end{array}
\right] \right\}  \\
&=&\E\left\{ \left[ 
\begin{array}{c}
\E\left\{ \left. S\left( A|W,X\right) \times \left\{ \frac{A\ }{f\left(
A|W,X\right) }-1\right\} \right\vert W,X\right\} \Omega(W,A,X) n\left( W,A,X\right)  \\ 
-\E\left\{ \left. \left\{ q\left( Z,A,X\right) -\frac{1}{f\left( A|W,X\right) 
}\right\} ^{2}\right\vert W,A,X\right\} n_{eff}\left( W,A,X\right) n\left(
W,A,X\right)  \\ 
-\left[ 
\begin{array}{c}
\left\{ \E\left( \left. n_{eff}\left( W,A,X\right) \Omega(W,A,X) \right\vert W,X\right) \right\} \\ 
\times \E\left\{ \left. \left\{ \frac{A\ }{f\left( A|W,X\right) }-1\right\}
^{2}\right\vert W,X\right\} \ \left\{ \Omega(W,A,X) \right\} 
\end{array}
\right] n\left( W,A,X\right) 
\end{array}
\right] \right\}, 
\end{eqnarray*}
where $\Omega(W,A,X)$ is shorthand for 
\begin{align*}
\left[ A\frac{1\ }{f\left( 1|W,X\right) }%
-\left( 1-A\right) \frac{f\left( 1|W,X\right) }{\left[ 1-f\left(
1|W,X\right) \right]^{2}}\right].
\end{align*}
This implies that
\[
0=\left[ 
\begin{array}{c}
\E\left\{ \left. S\left( A|W,X\right) \times \left\{ \frac{A\ }{f\left(
A|W,X\right) }-1\right\} \right\vert W,X\right\} \Omega(W,A,X) \\ 
-\E\left\{ \left. \left\{ q\left( Z,A,X\right) -\frac{1}{f\left( A|W,X\right) 
}\right\} ^{2}\right\vert W,A,X\right\} n_{eff}\left( W,A,X\right) \\ 
-\left[ 
\begin{array}{c}
\left\{ \E\left( \left. n_{eff}\left( W,A,X\right) \Omega(W,A,X) \right\vert W,X\right) \right\} 
\\ 
\times \E\left\{ \left. \left\{ \frac{A}{f\left( A|W,X\right) }-1\right\}
^{2}\right\vert W,X\right\}  \left\{ \Omega(W,A,X)\right\} 
\end{array}
\right] \ 
\end{array}
\right]. 
\]
So we further have that 
\[
0=\left[ 
\begin{array}{c}
\E\left\{ \left. S\left( A|W,X\right) \times \left\{ \frac{A}{f\left(
A|W,X\right) }-1\right\} \right\vert W,X\right\}\frac{\Omega^2(W,A,X)}{
\E\left\{ \left. \left\{ q\left( Z,A,X\right) -\frac{1}{f\left( A|W,X\right) }
\right\} ^{2}\right\vert W,A,X\right\} } \\ 
- n_{eff}\left(W,A,X\right) \Omega(W,A,X)\\ 
-\left[ 
\begin{array}{c}
\left\{ \E\left( \left. n_{eff}\left( W,A,X\right) \Omega(W,A,X) \right\vert W,X\right) \right\} 
\\ 
\times \E\left\{ \left. \left\{ \frac{A}{f\left( A|W,X\right) }-1\right\}
^{2}\right\vert W,X\right\}\\
\times \frac{\Omega^2(W,A,X)}{
\E\left\{ \left. \left\{ q\left( Z,A,X\right) -\frac{1}{f\left( A|W,X\right) }
\right\} ^{2}\right\vert W,A,X\right\} }
\end{array}
\right] \ 
\end{array}
\right] 
\]
and therefore
\[
0=\left[ 
\begin{array}{c}
\E\left\{ \left. S\left( A|W,X\right) \times \left\{ \frac{A\ }{f\left(
A|W,X\right) }-1\right\} \right\vert W,X\right\} \E\left[ \left. \frac{\Omega^2(W,A,X)}{
\E\left\{ \left. \left\{ q\left( Z,A,X\right) -\frac{1}{f\left( A|W,X\right) }
\right\} ^{2}\right\vert W,A,X\right\} }  \right\vert W,X\right]  \\ 
-\left\{ \E\left( \left. n_{eff}\left( W,A,X\right) \Omega(W,A,X) \right\vert W,X\right) \right\} \\ 
-\left[ 
\begin{array}{c}
\left\{ \E\left( \left. n_{eff}\left( W,A,X\right) \Omega(W,A,X) \right\vert W,X\right) \right\} 
\\ 
\times \E\left\{ \left. \left\{ \frac{A\ }{f\left( A|W,X\right) }-1\right\}
^{2}\right\vert W,X\right\} \E\left[ \left. \frac{ \Omega^2(W,A,X) }{\E\left\{ \left.
\left\{ q\left( Z,A,X\right) -\frac{1}{f\left( A|W,X\right) }\right\}
^{2}\right\vert W,A,X\right\} }\right\vert W,X\right] 
\end{array}
\right]
\end{array}
\right],
\]
which implies that 
\begin{eqnarray*}
&&\left\{ \E\left( \left. n_{eff}\left( W,A,X\right) \Omega(W,A,X) \right\vert W,X\right) \right\} \\
&=&\frac{\E\left\{ \left. S\left( A|W,X\right) \times \left\{ \frac{A}{
f\left( A|W,X\right) }-1\right\} \right\vert W,X\right\} \E\left[ \left. \frac{\Omega^2(W,A,X)}{\E\left\{ \left. \left\{ q\left( Z,A,X\right) -\frac{1}{f\left(
A|W,X\right) }\right\} ^{2}\right\vert W,A,X\right\} } \right\vert W,X
\right] }{\E\left\{ \left. \left\{ \frac{A}{f\left( A|W,X\right) }
-1\right\} ^{2}\right\vert W,X\right\} E\left[ \left. \frac{\Omega^2(W,A,X)}{\E\left\{ \left. \left\{ q\left( Z,A,X\right) -\frac{1}{f\left(
A|W,X\right) }\right\} ^{2}\right\vert W,A,X\right\} }\right\vert W,X\right]
+1} \\
&\equiv &v\left( W,X\right). 
\end{eqnarray*}
Therefore,
\begin{eqnarray*}
0 &=&\left[ 
\begin{array}{c}
\E\left\{ \left. S\left( A|W,X\right) \times \left\{ \frac{A}{f\left(
A|W,X\right) }-1\right\} \right\vert W,X\right\} \frac{\Omega^2(W,A,X)}{
\E\left\{ \left. \left\{ q\left( Z,A,X\right) -\frac{1}{f\left( A|W,X\right) }
\right\} ^{2}\right\vert W,A,X\right\} }  \\ 
-n_{eff}\left( W,A,X\right) \Omega(W,A,X) \\ 
-\left[ 
\begin{array}{c}
v(W,X)
\times \E\left\{ \left. \left\{ \frac{A}{f\left( A|W,X\right) }-1\right\}
^{2}\right\vert W,X\right\} \frac{\Omega^2(W,A,X)}{
\E\left\{ \left. \left\{ q\left( Z,A,X\right) -\frac{1}{f\left( A|W,X\right) }
\right\} ^{2}\right\vert W,A,X\right\} }
\end{array}
\right] \ 
\end{array}
\right].
\end{eqnarray*}
So we have  
\[
n_{eff}\left( W,A,X\right) \ =\left\{ 
\begin{array}{c}
\E\left\{ \left. S\left( A|W,X\right) \times \left\{ \frac{A}{f\left(
A|W,X\right) }-1\right\} \right\vert W,X\right\}\frac{\Omega(W,A,X)}{
\E\left\{ \left. \left\{ q\left( Z,A,X\right) -\frac{1}{f\left( A|W,X\right) }
\right\} ^{2}\right\vert W,A,X\right\} }  \\ 
-\left[ 
\begin{array}{c}
v\left(W,X\right) 
\times \E\left\{ \left. \left\{ \frac{A}{f\left( A|W,X\right) }-1\right\}
^{2}\right\vert W,X\right\} \frac{\Omega(W,A,X)}{
\E\left\{ \left. \left\{ q\left( Z,A,X\right) -\frac{1}{f\left( A|W,X\right) }
\right\} ^{2}\right\vert W,A,X\right\} }
\end{array}
\right] 
\end{array}
\right\}.
\]
\end{proof}

\section{Proof of Theorem~\ref{thm:estimator}}\label{sec:F}
We prove Theorem~\ref{thm:estimator} under the following regularity conditions (Appendix B of \cite{robinsetal1994}):\\
 Let $H(\phi)'= (EIF(\psi), S(\zeta)')$ and $\phi'= (\psi, \zeta')$, where $\zeta$ is a vector including all nuisance parameters, and $S(\zeta)$ are the estimating equations for solving $\zeta$.\\
 1. $\phi$  lies in the interior of a compact set;\\
 2. $\E[H(\phi)]\neq 0$ if $\phi\neq \phi_0$;\\
 3. $var[H(\phi_0)]$ is finite and positive definite;\\
 4. $\E[\partial H(\phi_0)/\partial \phi']$ exists and is invertible;\\
 5. A neighborhood $N$ of $\phi_0$ such that $\E[\sup_{\phi\in N}||H(\phi)||]$,
 $\E[\sup_{\phi\in N}||\partial H(\phi)/\partial \phi'||]$, and $\E[\sup_{\phi\in N}||H(\phi)H(\phi)'||]$ are all finite, where $||\cdot||$ denotes Frobenius norm;\\
 6. For all $\phi$ in a neighborhood $N$ of $\phi_0$, $\E_{\phi^*}[||H(\phi^*)||]$ and $\E_{\phi\in N}[H(\phi)H(\phi)']$ are bounded.

\begin{proof}
We start from the proof of double robustness. Under some regularity conditions \citep{white1982maximum}, the nuisance estimators, $\widehat q(Z,A,X)$ and $\widehat h(W,A,X)$ converge in probability to $q^*(Z,A,X)$ and $h^*(W,A,X)$. Suppose $h^*(W,A,X)$ is correctly specified,
\begin{align*}
&\E \left[(-1)^{1-A} q^*(Z,A,X) [Y- h(W,A,X)] + h(W,1,X) - h(W,0,X) - \psi \right]\\
= & \E \big[(-1)^{1-A} q^*(Z,A,X) [\E(Y|Z,A,X)- \int h(w,A,X)dF(w|Z,A,X)]\\
 & + h(W,1,X) - h(W,0,X)\big] - \psi  \\
= & \E[h(W,1,X) - h(W,0,X)] - \psi \\
= & 0.
\end{align*}
On the other hand, suppose $q^*(Z,A,X)$ is correctly specified,
\begin{align*}
&\E \left[(-1)^{1-A} q(Z,A,X) [Y- h^*(W,A,X)] + h^*(W,1,X) - h^*(W,0,X) - \psi \right]\\
= & \E \left[(-1)^{1-A} q(Z,A,X) Y |Z,A,X \right] - \E \left[(-1)^{1-A} q(Z,A,X) h^*(W,A,X) |Z,A,X \right]\\
+& \E\left[ h^*(W,1,X) - h^*(W,0,X) \right] - \psi.\\
\end{align*}

Note that
\begin{align*}
& \E \left[(-1)^{1-A} q(Z,A,X) Y\right]\\
= & \E \left[(-1)^{1-A} q(Z,A,X) \E[Y |Z,A,X] \right] \\
= & \E \left[(-1)^{1-A} q(Z,A,X) \E[h(W,A,X) |Z,A,X] \right]\\
= & \E \left[(-1)^{1-A} q(Z,A,X) h(W,A,X) \right]\\
= & \E \left[(-1)^{1-A}\E[q(Z,A,X)|W,A,X] h(W,A,X) \right] \\
= & \E \left[ \frac{(-1)^{1-A} h(W,A,X)}{f(A|W,X)} \right]\\
= & \psi.
\end{align*}

Furthermore,
\begin{align*}
& \E[ (-1)^{1-A}q(Z,A,X)h^*(W,A,X) ] \\
= & \E[ (-1)^{1-A} \E[q(Z,A,X)|W,A,X]h^*(W,A,X) ] \\
= & \E\left[  \frac{(-1)^{1-A}}{f(A|W,X)} h^*(W,A,X) \right]\\
= & \E[h^*(W,1,X) - h^*(W,0,X)]
\end{align*}
cancels with $\E\left[ h^*(W,1,X) - h^*(W,0,X)\right]$.

In order to show asymptotic normality and local efficiency, we need to derive the influence function of $\widehat \psi$.
Let $\zeta$ be a vector including all nuisance parameters.
 From a standard Taylor expansion of $EIF(\psi)$ around $\psi$ and $\zeta$, following uniform weak law of large number \citep{NEWEY19942111} under the regularity conditions, 
 we have
\begin{align*}
\sqrt n (\widehat \psi -\psi ) = \frac{1}{\sqrt n} \sum_{i=1}^{n} EIF(\psi; \cO_i) + \E( \frac{\partial EIF(\psi; \cO)}{\partial \zeta} )\sqrt n(\widehat \zeta - \zeta) + o_p(1).
\end{align*}
Following from the proof of double robustness, we have $\E[\partial EIF(\psi)/\partial \zeta]=0$ under the intersection model $\cM_{int}$.
This completes our proof.
\end{proof}

\section{Average treatment effect on the treated}\label{sec:G}

In this section, we briefly consider proximal inference about the average treatment effect on the treated, 
\begin{align*}
\mu  = \E[Y(1)-Y(0)|A=1], 
\end{align*}
 which is often equally of interest \citep{heckman1998characterizing,hahn1998role}.
As the first term in the contrast defining the ATT is point identified under consistency only, in order to identify the ATT, one solely needs to invoke the confounding bridge function for the control group only (along with  Assumptions~\ref{asm:consistency}, \ref{asm:condindy}-\ref{asm:condindw}, 
$Y (0) \perp A|U,X$, 
and the support of $\Pr(U,X|A=1)$ is a subset of the support of $\Pr(U,X|A=0)$), a somewhat weaker requirement than for the ATE. 
In the following theorem, we derive the efficient influence function of ATT and thus its corresponding semiparametric efficiency bound under the semiparametric model $\cM_{sp}$ which places no restrictions on the observed data distribution other than existence (but not necessarily uniqueness) of a bridge function $h$ that solves
\begin{align}
\E[Y|Z,A=0,X] &= \int h(w,X)dF(w|Z,A=0,X). \label{eq:attcondexpec}
\end{align}
A regularity condition is provided in Section~\ref{sec:G1}. 

\begin{theorem}\label{thm:eif2}
The efficient influence function of $\mu$ at the submodel of $\cM_{sp}$ where Section~\ref{sec:G1} holds, 
\begin{align}
\E[q(Z,X)|W, A=0, X] &= \frac{f(A=1|W, X)}{f(A=0|W, X)},\label{eq:attcondexpect2}
\end{align}
holds at the true data law, and $h$, $q$ are unique is given by
\begin{align}\label{eq:eif2}
EIF(\mu) = A Y/f(A=1) - (1-A) q(Z,X) [Y- h(W,X)]/f(A=1)  -A [ h(W,X) + \mu]/f(A=1).
\end{align}
Therefore, the semiparametric efficiency bound of $\mu$ at the submodel of $\cM_{sp}$ where Section~\ref{sec:G1} holds, Equation~\eqref{eq:attcondexpect2} holds at the true data law, and $h$, $q$ are unique equals $\E[EIF^2(\mu)]$.
\end{theorem}

The above theorem characterizes the efficient influence function of the treatment effect on the treated $\mu$ in $\cM_{sp}$ and therefore characterizes the semiparametric efficiency bound for the model.
Estimation and inference follows directly in a manner analogous to the approach described in the previous section for the ATE now based on the efficient influence function \eqref{eq:eif2}, although details are omitted.

\subsection{Regularity condition for average treatment effect on the treated}\label{sec:G1}






 Let $T:L_2(W,X)\rightarrow L_2(Z,X)$ be the operator given by $T(g)\equiv \E[g(W,X)|Z,A=0,X]$, and the adjoint $T':L_2(Z,X)\rightarrow L_2(W,X)$ be $T'(g)\equiv \E[g(Z,X)|W,A=0,X]$. We assume that $T$ and $T'$ are surjective.

\subsection{Proof of Theorem~\ref{thm:eif2}}\label{sec:G2}
\begin{proof}
We essentially need to consider $\mu = \E[Y(0)|A=1]$. Note that $\mu$ is identified by 
$$\E[h(W,X)|A=1]~~~~\text{and}~~~~\E[I(A=0)Yq(Z,X)]/p,$$
where $p$ stands for $f(A=1)$. The proofs of identification are similar to those of Theorems~\ref{thm:id} and \ref{thm:identification}, respectively, and thus we omit here. 
In order to find the efficient influence function for $\mu$, we need to first find a random variable $G$ with mean 0 and
\begin{align}\label{eq:score2}
\frac{\partial \mu_t}{\partial t}|_{t=0} = \E[GS(\cO;t)]|_{t=0},
\end{align}
where $S(\cO;t)=\partial \log f(\cO;t)/\partial t$, and $\mu_t$ is the parameter of interest $\mu$ under a regular parametric submodel in $\mM_{sp}$ indexed by $t$ that includes the true data generating mechanism at $t = 0$ \citep{vaart_1998}.

Recall that
\begin{align*}
\E[Y|Z,A=0,X] = \int h(w,X)dF(w|Z,A=0,X),
\end{align*}
so we have
\begin{align*}
\partial \E_t[Y-h_t(W,X)|Z,A=0,X]/\partial t|_{t=0}=0.
\end{align*}
Thus,
\begin{align*}
\int \frac{\partial [\{y-h_t(w,X)\}f_t(w,y|Z,A=0,X)]}{\partial t}|_{t=0} d(w,y) = 0.
\end{align*}
Let $\epsilon = Y - h(W,X)$, we have that
\begin{align*}
\E[\epsilon S(W,Y|Z,A=0,X)|Z,A=0,X] = \E[\partial h_t(W,X)/\partial t|_{t=0}|Z,A=0,X].
\end{align*}

The left hand side of Equation~\eqref{eq:score2} is equal to
\begin{align*}
&\partial \mu_t/\partial t|_{t=0}\\ =& \partial [\int \int h_t(w,x) dF_t(w,x|A=1)dF_t(x|A=1)]/\partial t|_{t=0} \\
=& \E[h(W,X)S(W,X|A=1)|A=1] +\int \int \partial h_t(w,x)/\partial t|_{t=0} dF(w,x|A=1)dF(x|A=1).
\end{align*}
The first term is equal to
\begin{align*}
& \E[h(W,X)S(W,X|A=1)|A=1] \\
=& \E[(h(W,X)-\mu)S(W,X|A=1)|A=1] \\
=& \E[A(h(W,X)-\mu)S(W,X|A)]/p \\
= & \E[A(h(W,X)-\mu)S(\cO)]/p.
\end{align*}
The second term is equal to
\begin{align*}
& \int \int \partial h_t(w,x)/\partial t|_{t=0} p dF(w,x|A=1)dF(x|A=1) \\
= &  \E[ \frac{f(A=0|W,X)}{f(A=0|W,X)}f(A=1|W,X)\partial h_t(W,X)/\partial t|_{t=0}]\\
= & \E[f(A=0|W,X)\E[q(Z,X)|W,A=0,X]\partial h_t(W,X)/\partial t|_{t=0} ]\\
= & \E[I(A=0)q(Z,X)\partial h_t(W,X)/\partial t|_{t=0} ]\\
= & \E[f(A=0|Z,X)q(Z,X)\E[\partial h_t(W,X)/\partial t|_{t=0}|Z,A=0,X] ]\\
= & \E[f(A=0|Z,X)q(Z,X)\epsilon S(W,Y|Z,A=0,X) ]\\
= & \E[I(A=0)q(Z,X)\epsilon S(W,Y|Z,A,X) ]\\
= & \E[I(A=0)q(Z,X)\epsilon S(W,Y|Z,A,X) ] + \E[I(A=0)q(Z,X)\epsilon S(Z,A,X) ]\\
= & \E[I(A=0)q(Z,X)\epsilon S(\cO) ].
\end{align*}
Finally, combing two terms gives
\begin{align*}
\frac{\partial \mu_t}{\partial t}|_{t=0} =\E\left[ \left\{ I(A=0) q(Z,X) [Y- h(W,X)]/p + I\{A=1\}[h(W,X)- \mu]/p  \right\} S(\cO)\right].
\end{align*}
Therefore, \begin{align}\label{eq:iff}
I\{A=0\} q(Z,X) [Y- h(W,X)]/p + I\{A=1\}[h(W,X)- \mu]/p,
\end{align}
is an influence function of $\mu$.  
Next, we show that the influence function~\eqref{eq:iff}
belongs to the tangent space 
\begin{align*}
&   \Lambda_1 + \Lambda_2 \\ & \equiv  \{S(Z,A,X) \in L_2(Z,A,X):\E[S(Z,A,X)]=0\}\\& + \{S(Y,W|Z,A,X)\in L_2(Z,A,X)^\perp:\E[\epsilon S(Y,W|Z,A=0,X)|Z,A=0,X]\in cl(R(T)), \\& ~~~~ S(Y,W|Z,A=1,X)~\text{unrestricted}\}.
\end{align*}
To see this, 
note that we have the following decomposition of Equation~\eqref{eq:iff}, 
\begin{align*}
& I\{A=0\} q(Z,X) [Y- h(W,X)]/p + I\{A=1\}\{h(W,X)- \mu\}/p\\
= & \{ \E[ I\{A=1\}\{h(W,X)- \mu\} |Z,A,X]\}/p\\&  + I\{A=1\} \{h(W,X)-\mu\}/p-\{\E[ I\{A=1\}\{h(W,X)-\mu\}|Z,A,X]\}/p\\& + I\{A=0\} q(Z,X) [Y- h(W,X)]/p,
\end{align*}
where 
$ \{\E[I\{A=1\}(h(W,X)- \mu)|Z,A,X]\}/p \in \Lambda_1$, and the remaining part 
belongs to $\Lambda_2$ as
\begin{align*}
  \E[  q(Z,X)\epsilon^2|Z,A=0,X]/p \in cl(R(T)),
\end{align*}
and
\begin{align*}
&\E[I\{A=1\} \{h(W,X)-\mu\}/p -\{\E[ I\{A=1\}\{h(W,X)-\mu\}|Z,A,X]\}/p]=0,\\
&\E[I\{A=0\} q(Z,X)\epsilon/p|Z,A,X]=0,
\end{align*}
which therefore completes the proof.
\end{proof}

\section{Choices of the parameters for data generating process}\label{sec:H}

We consider generating data $(X,W,A,Z,U,Y)\ $such that
\begin{eqnarray*}
\E\left( Y|U,a,X\right)  &=& \int h(w,a,X)dF(w|U,X), a=0,1, \\
\frac{1}{\Pr \left( A=a|U,X\right) } &=&\int q(z,a,X)dF(z|U,a,X),a=0,1.
\end{eqnarray*}
Let
\[
h(W,A,X)=b_{0}+b_{a}A+b_{w}W+b_{x}X,
\]
and 
\[
q(Z,A,X)=1+\exp \left\{ \left( -1\right) ^{1-A}t_{0}+\left( -1\right)
^{1-A}t_{z}Z+\left( -1\right) ^{1-A}t_{a}A + \left( -1\right) ^{1-A}t_{x}X\right\},
\]
therefore, we have that 
\begin{eqnarray*}
&&\frac{1}{\Pr \left(A|U,X\right) }\\&=&1+\exp \left\{ t_{0}\left( -1\right)
^{1-A}+t_{a}A\left( -1\right) ^{1-A}+t_{x}X\left( -1\right) ^{1-A}\right\} \int \exp \left\{ \left(
-1\right) ^{1-A}t_{z}z\right\} dF(z|U,A,X).
\end{eqnarray*}
Suppose that $Z|U,A,X\sim N(\theta _{0}+\theta _{a}A+\theta _{u}U+\theta _{x}X,\sigma
_{z|u,a,x}^{2})$ so that
\begin{eqnarray*}
&&\frac{1}{\Pr \left( A|U,X\right) }\\ &=&1+\exp \left\{ \left( -1\right)
^{1-A}t_{0}+\left( -1\right) ^{1-A}t_{a}A+\left( -1\right) ^{1-A}t_{x}X\right\} \int \exp \left\{ \left(
-1\right) ^{1-A}t_{z}z\right\} dF(z|U,A,X) \\
&=&1+\exp \left\{ \left( -1\right) ^{1-A}[t_{0}+t_{a}A+t_{x}X]+\left( -1\right) ^{1-A}t_{z}\left( \theta _{0}+\theta
_{a}A+\theta _{u}U+\theta _{x}X\right) +\frac{t_{z}^{2}\sigma _{z|u,a,x}^{2}}{2}\right\}.
\end{eqnarray*}
Now we need
\begin{eqnarray*}
\Pr \left( A=1|U,X\right) +\Pr \left( A=0|U,X\right)  = 1.
\end{eqnarray*}
Thus
\begin{eqnarray*}
&&t_{0}+t_{a}+t_{x}X+t_{z}\left( \theta _{0}+\theta _{a}+\theta
_{u}U+\theta _{x}X\right) +\frac{t_{z}^{2}\sigma _{z|u,a,x}^{2}}{2} \\
&=&t_{0}+t_{x}X+t_{z}\left( \theta _{0}+\theta _{u}U+\theta _{x}X\right) -\frac{
t_{z}^{2}\sigma _{z|u,a,x}^{2}}{2},
\end{eqnarray*}
which implies that
\[
t_{a}=-t_{z}^{2}\sigma _{z|u,a,x}^{2}-t_{z}\theta_a.
\]
We therefore have that
\begin{eqnarray*}
q(Z,A,X)=1+\exp \left\{ \left( -1\right)^{1-A}[t_{0}+t_{z}Z+t_{x}X-t_{z}^{2}\sigma _{z|u,a,x}^{2}A-t_{z}\theta_a
A]\right\}.
\end{eqnarray*}

Next we let
\[
\left( Z,W,U\right) |A,X\sim MVN\left( \left(
\begin{array}{c}
\alpha _{0}+\alpha _{a}A+\alpha _{x}X \\
\mu _{0}+\mu _{a}A+\mu _{x}X \\
\kappa _{0}+\kappa _{a}A+\kappa _{x}X
\end{array}
\right) ,\left(
\begin{array}{ccc}
\sigma _{z}^{2} & \sigma _{zw} & \sigma _{zu} \\
\sigma _{zw} & \sigma _{w}^{2} & \sigma _{wu} \\
\sigma _{zu} & \sigma _{wu} & \sigma _{u}^{2}
\end{array}
\right) \right).
\]
Therefore, 
\begin{eqnarray*}
\E[Z|U,A,X] = \alpha _{0}+\alpha _{a}A+\alpha _{x}X+\frac{\sigma_{zu}}{\sigma_{u}^2} (U-\kappa_0-\kappa_a A-\kappa_x X),\\
\theta_0 = \alpha_0-\frac{\sigma_{zu}}{\sigma_{u}^2}\kappa_0,\\
\theta_a = \alpha_a -\frac{\sigma_{zu}}{\sigma_{u}^2}\kappa_a,\\
\theta_x = \alpha_x -\frac{\sigma_{zu}}{\sigma_{u}^2}\kappa_x,\\
\theta_u=\frac{\sigma_{zu}}{\sigma_{u}^2}.
\end{eqnarray*}

In addition, we impose
\begin{eqnarray*}
&&W\perp (A,Z)|U,X.
\end{eqnarray*}

The independence implies that $W|U,A,Z,X \sim$
\begin{eqnarray*}
&& N\left( \left( \mu _{0}+\mu _{a}A+\mu _{x}X+\Sigma _{w(u,z)}\Sigma
_{u,z}^{-1}\left(
\begin{array}{c}
U-\kappa _{0}-\kappa _{a}A-\kappa _{x}X\\
Z-\alpha _{0}-\alpha _{a}A-\alpha _{x}X
\end{array}
\right) \right) ,\sigma _{w}^{2}-\Sigma _{w(u,z)}\Sigma _{u,z}^{-1}\Sigma
_{w(u,z)}^{T}\right), \\
&&\text{where}\\
&& \Sigma _{w(u,z)} = \left(
\begin{array}{cc}
\sigma _{wu} & \sigma _{wz}
\end{array}
\right),\\
&& \Sigma _{u,z} =\left(
\begin{array}{cc}
\sigma _{u}^{2} & \sigma _{zu}\\
\sigma _{zu} & \sigma _{z}^{2}
\end{array}
\right),
\end{eqnarray*}
such that
\[
\E\left( W|U,A,Z,X\right) =\E\left( W|U,A,X\right) =\mu _{0}+\mu _{a}A+\mu _{x}X+\frac{
\sigma _{wu}}{\sigma _{u}^{2}}\left( U-\left( \kappa _{0}+\kappa
_{a}A+\kappa
_{x}X\right) \right)
\]%
does not depend on $A$ and $Z$. Therefore
\[
\frac{\sigma_{wz}\sigma_{u}^2-\sigma_{wu}\sigma_{zu}}{\sigma_{z}^2\sigma_{u}^2-\sigma_{zu}^2}=0.
\]
and
\[
\mu _{a}=\frac{\sigma _{wu}}{\sigma _{u}^{2}}\kappa _{a}.
\]

Moreover, we impose that
\[
Z\perp Y|U,A,X,
\]
which we impose by setting
\begin{eqnarray*}
\E\left( Y|W,U,A,Z,X\right)  &=&\E\left( Y|U,A,Z,X\right) +\omega \left\{
W-\E\left( W|U,A,Z,X\right) \right\}  \\
&=&\E\left( Y|U,A,X\right) +\omega \left\{ W-\E\left( W|U,X\right) \right\}  \\
&=&b_{0}+b_{a}A+b_{x}X+b_{w}\E\left( W|U,X\right) +\omega \left\{ W-\E\left( W|U,X\right)
\right\}  \\
&=&b_{0}+b_{a}A+b_{x}X+\left( b_{w}-\omega \right) \E\left( W|U,X\right) +\omega W,
\end{eqnarray*}
where
\[
\E\left( W|U,X\right) =\E\left( W|U,A,Z,X\right) =\mu _{0}+\mu _{x}X+\frac{\sigma _{wu}}{
\sigma _{u}^{2}}\left( U-\kappa _{0}-\kappa_{x}X\right).
\]
Furthermore, notice that as $\Pr \left( A=a|U,X\right) =\Pr \left( A=a|U,W,X\right) $,
the log odds ratio must be that
\begin{eqnarray*}
\log OR\left( A,U|W,X\right)  &=&-t_{z}\theta _{u}A U\\
&=&\frac{
\E\left( U|A=1,W,X\right) -\E\left( U|A=0,W,X\right) }{\sigma _{u|w,a,x}^{2}}UA  \\
&=&\frac{\kappa _{a}-\sigma _{wu}\mu _{a}/\sigma _{w}^{2}}{\sigma
_{u|w,a,x}^{2}}UA.
\end{eqnarray*}
Therefore
\[
t_{z}\theta _{u}=-\frac{\kappa _{a}-\sigma _{wu}\mu _{a}/\sigma _{w}^{2}}{
\sigma _{u|w,a,x}^{2}},
\]
and
\[
t_{z}=-\frac{\kappa _{a}-\sigma _{wu}\mu _{a}/\sigma _{w}^{2}}{\theta
_{u}\sigma _{u|w,a,x}^{2}}.
\]

Recall that
\begin{eqnarray*}
&&\frac{1}{\Pr \left( A=a|W,X\right) }\\ &=&\int \frac{1}{\Pr \left(
A=a|U,W,X\right) }dF\left( U|W,A=a,X\right)  \\
&=&1+\exp \left\{ \left( -1\right) ^{1-a}[t_{0}+t_{a}a+t_{x}X]+\left( -1\right) ^{1-a}t_z\left( \theta _{0}+\theta _{a}a+\theta _{x}X\right) +\frac{t_{z}^{2}\sigma
_{z|u,a,x}^{2}}{2}\right\}  \\
&&\times \int \exp \left[ \left( -1\right) ^{1-a}t_{z}\theta _{u}U\right]
dF\left( U|W,A=a,X\right)  \\
&=&1+\exp \left\{ \left( -1\right) ^{1-a}[t_{0}+t_{a}a+t_{x}X]+\left( -1\right)
^{1-a}t_{z}\left( \theta _{0}+\theta _{a}a+\theta _{x}X\right) +\frac{t_{z}^{2}\sigma
_{z|u,a,x}^{2}}{2}\right\}  \\
&&\times \exp \left[ \left( -1\right) ^{1-a}t_{z}\theta
_{u}\E(U|W,A=a,X)+\sigma _{u|w,a,x}^{2}\frac{t_{z}^{2}\theta _{u}^{2}}{2}\right].
\end{eqnarray*}

Note that
\[
\Pr \left( A=0|W,X\right) +\Pr \left( A=1|W,X\right) =1,
\]
as long as
\begin{eqnarray*}
&&t_{0}+t_{a}+t_{x}X+t_{z}
\left( \theta _{0}+\theta _{a}+\theta _{x}X\right) +\frac{t_{z}^{2}\sigma _{z|u,a,x}^{2}}{2}
\\
&&+t_{z}\theta _{u}\left[\E\left( U|A=1,W,X\right) -\E\left( U|A=0,W,X\right)
\right] +\sigma _{u|w,a,x}^{2}\frac{t_{z}^{2}\theta _{u}^{2}}{2} \\
&=&t_{0}+t_{x}X+t_{z}(\theta _{0}+\theta_{x}X)-\frac{t_{z}^{2}\sigma _{z|u,a,x}^{2}}{2} -\sigma _{u|w,a,x}^{2}\frac{t_{z}^{2}\theta _{u}^{2}}{2},
\end{eqnarray*}
which holds under the model because%
\[
\E\left( U|A=0,W,X\right) -\E\left( U|A=1,W,X\right) =-[\kappa _{a}-\sigma
_{wu}\mu _{a}/\sigma _{w}^{2}],
\]
and
\[
t_{z}\theta _{u}=-\frac{\kappa _{a}-\sigma _{wu}\mu _{a}/\sigma _{w}^{2}}{
\sigma _{u|w,a,x}^{2}}.
\]
Thus
\[
\E\left( U|A=1,W,X\right) -\E\left( U|A=0,W,X\right) =-t_{z}\theta _{u}\sigma_{u|w,a,x}^2,
\]
and therefore
\[
t_{z}\theta _{u}\left[ \E\left( U|A=1,W,X\right) -\E\left( U|A=0,W,X\right)
\right] =-t_{z}^{2}\theta _{u}^{2}\sigma_{u|w,a,x}^2.
\]
Recall also that
\[
t_{a}=-t_{z}^{2}\sigma _{z|u,a,x}^{2}-t_z\theta_a,
\]
therefore
\begin{eqnarray*}
&&t_{a}+t_{z}
\theta _{a} +t_{z}^{2}\sigma _{z|u,a,x}^{2}-\sigma_{u|w,a,x}^2t_{z}^2\theta _{u}^2 +\sigma _{u|w,a,x}^{2}t_{z}^{2}\theta _{u}^{2}=0.
\end{eqnarray*}

Finally, recall that
\begin{eqnarray*}
&&\frac{1}{\Pr \left(A|U,X\right) }\\ &=&1+\exp \left\{ \left( -1\right)
^{1-A}t_{0}+\left( -1\right) ^{1-A}t_{a}A+\left( -1\right) ^{1-A}t_{x}X\right\} \int \exp \left\{ \left(
-1\right) ^{1-A}t_{z}Z\right\} dF(Z|U,A,X) \\
&=&1+\exp \left\{ \left( -1\right) ^{1-A}[t_{0}+t_{a}A+t_{x}X]+\left( -1\right) ^{1-A}t_{z}\left( \theta _{0}+\theta
_{a}A+\theta _{u}U+\theta _{x}X\right) +\frac{t_{z}^{2}\sigma _{z|u,a,x}^{2}}{2}\right\},
\end{eqnarray*}
and thus
\begin{eqnarray*}
&&\frac{1}{\Pr \left( A=a|X\right) }\\
&=& 1+\exp \left\{ \left( -1\right) ^{1-a}[t_{0}+t_{a}a+t_{x}X]+\left( -1\right) ^{1-a}t_{z}\left( \theta _{0}+\theta
_{a}a+\theta
_{x}X\right) +\frac{t_{z}^{2}\sigma _{z|u,a,x}^{2}}{2}\right\}\\
&&\int \exp \left[ \left( -1\right) ^{1-a}t_{z}\theta _{u}U\right] dF(U|A=a,X)\\
&=& 1+\exp \left\{ \left( -1\right) ^{1-a}[t_{0}+t_{a}a+t_{x}X]+\left( -1\right) ^{1-a}t_{z}\left( \theta _{0}+\theta
_{a}a+\theta
_{x}X\right) +\frac{t_{z}^{2}\sigma _{z|u,a,x}^{2}}{2}\right\}\\
&&\exp \left[ \left( -1\right) ^{1-a}t_{z}\theta
_{u}\E(U|A=a,X)+\sigma _{u|a,x}^{2}\frac{t_{z}^{2}\theta _{u}^{2}}{2}\right]\\
&=& 1+\exp \left\{ \left( -1\right) ^{1-a}[t_{0}+t_{a}a+t_{x}X]+\left( -1\right) ^{1-a}t_{z}\left( \theta _{0}+\theta
_{a}a+\theta
_{x}X\right) +\frac{t_{z}^{2}\sigma _{z|u,a,x}^{2}}{2}\right\}\\
&&\exp \left[ \left( -1\right) ^{1-a}t_{z}\theta
_{u}[\kappa_0+\kappa_a a+\kappa_x X]+\sigma _{u|a,x}^{2}\frac{t_{z}^{2}\theta _{u}^{2}}{2}\right].
\end{eqnarray*}

Thus, $A|X$ is generated by
\begin{eqnarray*}
&&1/\Pr(A=1|X)\\&=& 1+\exp \left\{ t_{0}+t_{a}+t_{x}X+t_{z}\left( \theta _{0}+\theta
_{a}+\theta
_{x}X\right) +\frac{t_{z}^{2}(1-\frac{\sigma_{zu}^2}{\sigma_z^2\sigma_u^2})\sigma _{z}^{2}}{2}\right\}\\
&&\exp \left[t_{z}\theta
_{u}[\kappa_0+\kappa_a+\kappa_x X]+\sigma _{u}^{2}\frac{t_{z}^{2}\theta _{u}^{2}}{2}\right].
\end{eqnarray*}

We conclude this section by summarizing the constraints of data generating mechanism, 
\begin{align*}
t_{a}=-t_{z}^{2}\sigma _{z|u,a,x}^{2}-t_{z}\theta_a=-t_{z}^{2}[(1-\frac{\sigma_{zu}^2}{\sigma_z^2\sigma_u^2})\sigma_{z}^2]-t_{z}\theta_a,\\
\sigma_{wz}\sigma_{u}^2-\sigma_{wu}\sigma_{zu}=0,\\
\mu _{a}\sigma _{u}^{2}=\sigma _{wu}\kappa _{a},\\
-\theta
_{u}t_{z}[(1-\frac{\sigma_{wu}^2}{\sigma_u^2\sigma_w^2})\sigma_{u}^2]=-\theta
_{u}\sigma _{u|w,a,x}^{2}t_{z}=\kappa _{a}-\sigma _{wu}\mu _{a}/\sigma _{w}^{2}.
\end{align*}

\section{Additional numerical results}\label{sec:addi}

\subsection{Simulation results when $U$ is not a confounder}

Consider $\kappa_a=0$, so $U$ does not affect $A$ and therefore is not a confounder. The parameters are set as follows:
\begin{itemize}
\setlength\itemsep{1em}

\item $\Gamma_x=(0.25,0.25)^T$, $\Sigma_x=\left(
\begin{array}{ccc}
\sigma_x^2 & 0\\
0 & \sigma_x^2\\
\end{array}
\right)$, $\sigma_x=0.25$.

\item $\Pr \left(A=1|X\right)=\left[1+ \exp\{(0.125,0.125)^TX\}\right]^{-1}$.

\item $\alpha_0= 0.25$, $\alpha_a= 0.25$, $\alpha_x= (0.25,0.25)^T$.

\item $\mu_0= 0.25$, $\mu_a= 0$, $\mu_x= (0.25,0.25)^T$.

\item $\kappa_0= 0.25$, $\kappa_a= 0$, $\kappa_x= (0.25,0.25)^T$.

\item $\Sigma=\left(
\begin{array}{ccc}
1 & 0.25 & 0.5 \\
0.25 & 1 & 0.5 \\
0.5 & 0.5 & 1
\end{array}
\right), \sigma_y=0.25.$

\item $b_0= 2$, $b_a= 2$, $b_x= (0.25,0.25)^T$, $b_w=4$, $\omega=2$. 

\item $t_0=0.25, t_z=0, t_a=0,$ $t_x= (0.25,0.25)^T$.
\end{itemize}

\begin{table}[!h]
\begin{center}
\caption{\label{j1}
Simulation results: absolute bias ($\times 10^{-2}$) and MSE ($\times 10^{-2}$)}
\begin{tabular}{cccccc}
\noalign{\smallskip}
\noalign{\smallskip}
 & & $\widehat \psi_{DR}$ & $\widehat \psi_{POR}$ & $\widehat \psi_{PIPW}$ & $\widehat \psi_{PDR}$ \\
\noalign{\smallskip}
\multirow{2}{*}{Scenario~1}  & Bias &  0.7  &  0.2  &  0.2  & 0.2 \\
                          & MSE &  1.5  &   0.7  &  0.7 & 0.7   \\
\noalign{\smallskip}
\multirow{2}{*}{Scenario~2} & Bias  &  0.7 &  37.3  &  0.2  & 0.7 \\
                & MSE &  1.5   &  21.4  &  0.7 & 3.6 \\
\noalign{\smallskip}
\multirow{2}{*}{Scenario~3} & Bias  & 0.7 &  0.2  &  0.4  & 0.2  \\
                & MSE  &  1.5  &  0.7  & 0.2 & 0.7  \\
                \noalign{\smallskip}
\multirow{2}{*}{Scenario~4} & Bias  & 0.7 &  21.8  &  0.3  & 0.1  \\
                & MSE  &  1.5  &  17.5  &  0.4 & 10.8 \\
\noalign{\smallskip}
\end{tabular}
\end{center}
\end{table}

\begin{table}[!h]
\begin{center}
\caption{\label{j2}
Simulation results: coverage ($\%$) and average length ($\times 10^{-2}$)}
\begin{tabular}{cccccc}
\noalign{\smallskip}
\noalign{\smallskip}
 & & $\widehat \psi_{DR}$ & $\widehat \psi_{POR}$ & $\widehat \psi_{PIPW}$ & $\widehat \psi_{PDR}$ \\
\noalign{\smallskip}
\multirow{2}{*}{Scenario~1}  & Coverage &  94.6   &  94.2  &  94.2  &  94.2  \\
                          & Length &  47.0  &   31.2  &  31.2 & 31.2   \\
\noalign{\smallskip}
\multirow{2}{*}{Scenario~2} & Coverage  &  94.6 &  41.8  &  94.2  & 97.6  \\
                & Length &  47.0  &  71.2  &  31.2 & 69.5  \\
\noalign{\smallskip}
\multirow{2}{*}{Scenario~3} & Coverage  & 94.6 & 94.2  &  99.8  & 94.8  \\
                & Length  &  47.0  &  31.2  &  32.4 & 32.5  \\
                \noalign{\smallskip}
\multirow{2}{*}{Scenario~4} & Coverage  & 94.6 &  88.4  & 98.0  & 99.6  \\
                & Length  &  47.0  & 136.4  & 33.2 & 142.9 \\
\noalign{\smallskip}
\end{tabular}
\end{center}
\source{}
\end{table}

In this subsection, the standard doubly robust estimator is given by
\begin{align*}
\widehat \psi_{DR} = \PP_n \left\{ \frac{(-1)^{1-A}}{\widehat f(A|X)} \{Y- \widehat \E[Y|X,A]\} + \widehat \E[Y|X,A=1] -\widehat \E[Y|X,A=0] \right\},
\end{align*}
where $\widehat f(A|X)$ and $\widehat \E[Y|X,A]$ are estimated via standard logistic regression and linear regression, respectively.
As can be seen from Tables~\ref{j1} and \ref{j2}, as expected, both the standard doubly robust estimator and the proposed estimators perform well in this setting when the working models are correctly specified.

\subsection{Sensitivity analysis on violation of Assumptions~\ref{asm:condindy} and \ref{asm:condindw}}

For violation of Assumptions~\ref{asm:condindy} and \ref{asm:condindw}, we consider the following model similar to that of \cite{miao2018confounding} but with some modifications so that $Z$ affects $W$:
\begin{itemize}
\setlength\itemsep{1em}

\item $X,U \sim MVN\left( \left(
\begin{array}{c}
0\\
0
\end{array}
\right) ,\left(
\begin{array}{cc}
1 & 0.5 \\
0.5 & 1
\end{array}
\right) \right).$

\item $Z = 0.5 + 0.5X + U + \epsilon_1$,
$\logit\{\Pr(A = 1 | Z, X, U)\} = - 0.5 + Z + 0.5X + 0.3U$.

\item $W = 1 - X + 0.4 U + 1.5 Z +\epsilon_2$, $Y (a) = 1 + 0.5a + 2X + U + 1.5aU + 2\epsilon_2$.

\item $\epsilon_1, \epsilon_2 \sim N(0, 1)$.

\end{itemize}

\begin{table}[!h]
\begin{center}
\caption{\label{j3}
Simulation results: absolute bias ($\times 10^{-2}$) and MSE ($\times 10^{-2}$)}
\begin{tabular}{cccccc}
\noalign{\smallskip}
\noalign{\smallskip}
 & & $\widehat \psi_{DR}$ & $\widehat \psi_{POR}$ & $\widehat \psi_{PIPW}$ & $\widehat \psi_{PDR}$ \\
\noalign{\smallskip}
\multirow{2}{*}{Scenario~1}  & Bias &  10.3  &  18.4  &  20.5  &  20.6 \\
                          & MSE &  3.7  &   4.6  &  7.8 &  7.8   \\
\noalign{\smallskip}
\multirow{2}{*}{Scenario~2} & Bias  &  22.6 &  2.1  &  20.5  & 24.0 \\
                & MSE &  15.0   &  1.7  &  7.8  & 14.5 \\
\noalign{\smallskip}
\multirow{2}{*}{Scenario~3} & Bias  & 68.8 &  18.4  &  11.7  & 45.2  \\
                & MSE  & 50.2  &  4.6  &  53.4 & 37.0  \\
                \noalign{\smallskip}
\multirow{2}{*}{Scenario~4} & Bias  & 29.0 &  2.2  &  53.4  & 8.8  \\
                & MSE  &  10.8  &  1.7  &  117.0 & 46.4 \\
\noalign{\smallskip}
\end{tabular}
\end{center}
\end{table}

\begin{table}[!h]
\begin{center}
\caption{\label{j4}
Simulation results: coverage ($\%$) and average length ($\times 10^{-2}$)}
\begin{tabular}{cccccc}
\noalign{\smallskip}
\noalign{\smallskip}
 & & $\widehat \psi_{DR}$ & $\widehat \psi_{POR}$ & $\widehat \psi_{PIPW}$ & $\widehat \psi_{PDR}$ \\
\noalign{\smallskip}
\multirow{2}{*}{Scenario~1}  & Coverage &  79.2   &  66.2  &  73.6  &  74.0  \\
                          & Length &  50.7 &   45.6  &  65.2 &  65.3 \\
\noalign{\smallskip}
\multirow{2}{*}{Scenario~2} & Coverage  &  78.8 &  96.0  &   73.6  & 74.6  \\
                & Length & 101.9  &  51.6  &  65.2 & 94.1  \\
\noalign{\smallskip}
\multirow{2}{*}{Scenario~3} & Coverage  & 5.6 &  66.2  & 42.0  & 47.2  \\
                & Length  & 66.3  & 45.6  &  673.5 & 679.2  \\
                \noalign{\smallskip}
\multirow{2}{*}{Scenario~4} & Coverage  & 55.0 &  95.8  &  42.6  & 63.4  \\
                & Length  & 61.7  & 51.6  &  987.9 &  1915.2 \\
\noalign{\smallskip}
\end{tabular}
\end{center}
\source{}
\end{table}

As can be seen from Tables~\ref{j3} and \ref{j4}, the proximal estimators are comparable to the standard doubly robust estimator.
It is not surprising that the proximal estimators are invalid as the identifying assumptions are violated. 

\subsection{Sensitivity analysis on dependence between $Z$ and $W$}

In Section~\ref{sec:numeric}, the correlation coefficient between $Z$ and $W$ given $X$ and $A$ is $\sigma_{wz}=0.25$. In this section, we consider a weaker association between $Z$ and $W$ given $X$ and $A$, i.e., $\sigma_{wz}=0.15$. The parameters are set as follows:
\begin{itemize}
\setlength\itemsep{1em}

\item $\Gamma_x=(0.25,0.25)^T$, $\Sigma_x=\left(
\begin{array}{ccc}
\sigma_x^2 & 0\\
0 & \sigma_x^2\\
\end{array}
\right)$, $\sigma_x=0.25$.

\item $\Pr \left(A=1|X\right)=\left[1+ \exp\{(0.125,0.125)^TX\}\right]^{-1}$.

\item $\alpha_0= 0.25$, $\alpha_a= 0.25$, $\alpha_x= (0.25,0.25)^T$.

\item $\mu_0= 0.25$, $\mu_a= 0.075$, $\mu_x= (0.25,0.25)^T$.

\item $\kappa_0= 0.25$, $\kappa_a= 0.25$, $\kappa_x= (0.25,0.25)^T$.

\item $\Sigma=\left(
\begin{array}{ccc}
1 & 0.15 & 0.5 \\
0.15 & 1 & 0.3 \\
0.5 & 0.3 & 1
\end{array}
\right), \sigma_y=0.25.$

\item $b_0= 2$, $b_a= 2$, $b_x= (0.25,0.25)^T$, $b_w=4$, $\omega=2$. 

\item $t_0=0.25, t_z=-0.5, t_a=-0.125,$ $t_x= (0.25,0.25)^T$.
\end{itemize}

\begin{table}[!h]
\begin{center}
\caption{\label{j5}
Simulation results: absolute bias ($\times 10^{-2}$) and MSE ($\times 10^{-2}$)}
\begin{tabular}{cccccc}
\noalign{\smallskip}
\noalign{\smallskip}
 & & $\widehat \psi_{DR}$ & $\widehat \psi_{POR}$ & $\widehat \psi_{PIPW}$ & $\widehat \psi_{PDR}$ \\
\noalign{\smallskip}
\multirow{2}{*}{Scenario~1}  & Bias &  7.1  &  0.2  &  0.6  & 0.7 \\
                          & MSE &  0.6  &  0.8  &  0.9 & 1.0   \\
\noalign{\smallskip}
\multirow{2}{*}{Scenario~2} & Bias  &  14.3 &  48.0  &   0.6  & 0.0 \\
                & MSE &  2.9   &  25.7  &  0.9 & 1.2 \\
\noalign{\smallskip}
\multirow{2}{*}{Scenario~3} & Bias  & 12.7 &  0.2  &   16.6  & 0.0  \\
                & MSE  &  1.7  &  0.8  &  2.9 & 0.8  \\
                \noalign{\smallskip}
\multirow{2}{*}{Scenario~4} & Bias  & 26.2 &  34.2  &   13.8  & 20.0  \\
                & MSE  &  7.8  &  15.6  &  2.1 & 7.3 \\
\noalign{\smallskip}
\end{tabular}
\end{center}
\end{table}

\begin{table}[!h]
\begin{center}
\caption{\label{j6}
Simulation results: coverage ($\%$) and average length ($\times 10^{-2}$)}
\begin{tabular}{cccccc}
\noalign{\smallskip}
\noalign{\smallskip}
 & & $\widehat \psi_{DR}$ & $\widehat \psi_{POR}$ & $\widehat \psi_{PIPW}$ & $\widehat \psi_{PDR}$ \\
\noalign{\smallskip}
\multirow{2}{*}{Scenario~1}  & Coverage &  22.4   &  95.8  &  96.8  &  96.8  \\
                          & Length & 9.9  &   36.0  &  209.6 & 209.7  \\
\noalign{\smallskip}
\multirow{2}{*}{Scenario~2} & Coverage  & 67.2 &  13.0  &  96.8  & 98.4 \\
                & Length &  36.7  & 55.1  &  209.6  & 214.9  \\
\noalign{\smallskip}
\multirow{2}{*}{Scenario~3} & Coverage  & 0.8 & 95.8  &  57.6  & 95.6  \\
                & Length  &  10.8  &  36.0  &  44.9 & 46.1  \\
                \noalign{\smallskip}
\multirow{2}{*}{Scenario~4} & Coverage  & 20.2 &  56.8  &  77.0  & 88.8  \\
                & Length  &  37.7  &  87.7  &  48.9 & 102.3 \\
\noalign{\smallskip}
\end{tabular}
\end{center}
\source{}
\end{table}

As can be seen from Tables~\ref{j5} and \ref{j6}, the proximal estimators are either comparable to or slightly worse than that in Section~\ref{sec:numeric}, and they still outperform the standard doubly robust estimator in terms of bias and coverage. 

\subsection{Sensitivity analysis on real data application}

In this section, we conducted the sensitivity analysis by removing a variable from $Z$ and $W$ respectively in the data application. Table~\ref{real2} reports corresponding point estimates and confidence intervals. The results of Scenarios 1 and 2 where $Z$ only includes pafi1 do not change much. In contrast, Scenarios 3 and 4, where $Z$ only includes paco21, proximal OR and proximal DR estimates are somewhat smaller and proximal IPW estimate is positive (although not statistically significant) unlike other estimates.
This suggests that paco21 by itself may not be a sufficiently relevant treatment confounding proxy to completely account for confounding. 
In addition, the discrepancy between proximal estimators suggests potential model misspecification in this case.

\begin{table}[!h]
\begin{center}
\caption{\label{real2}
Treatment effect estimates (standard deviations) and 95\% confidence intervals of the average treatment effect. Scenario 1: $W$=ph1, $Z$=pafi1; Scenario 2: $W$=hema1, $Z$=pafi1; Scenario 3: $W$=ph1, $Z$=paco21; Scenario 4: $W$=hema1, $Z$=paco21.}
\begin{tabular}{cccccc}
\noalign{\smallskip}
\noalign{\smallskip}
& & $\widehat \psi_{DR}$ & $\widehat \psi_{POR}$ & $\widehat \psi_{PIPW}$ & $\widehat \psi_{PDR}$ \\
\noalign{\smallskip}
\multirow{2}{*}{1} & Treatment effects (SDs) &  -1.17 (0.32)   &  -1.92 (0.44)  &  -1.64 (0.46)  &  -1.74 (0.55)\\
& 95\% CIs &  (-1.79,-0.55)  &  (-2.78,-1.06)  &  (-2.54,-0.73) & (-2.82,-0.66)   \\
\noalign{\smallskip}
\multirow{2}{*}{2} & Treatment effects (SDs) &  -1.17 (0.32)   &  -1.73 (0.52)  &  -1.65 (0.31)  &   -1.61 (0.47)\\
& 95\% CIs &  (-1.79,-0.55)  &   (-2.75,-0.71)  &  (-2.25,-1.05) & (-2.54,-0.68)   \\
\noalign{\smallskip}
\multirow{2}{*}{3} & Treatment effects (SDs) &  -1.17 (0.32)   &  -1.35 (0.27)  &  0.41 (0.25)  &  -1.35 (0.27)\\
& 95\% CIs &  (-1.79,-0.55)  &   (-1.89,-0.81)  &  (-0.07,0.90) & (-2.50,-0.82)   \\
\noalign{\smallskip}
\multirow{2}{*}{4} & Treatment effects (SDs) &  -1.17 (0.32)   &  -1.02 (0.33)  &  0.37 (0.29)  &   -1.01 (0.34)\\
& 95\% CIs &  (-1.79,-0.55)  &   (-1.67,-0.36)  &  (-0.20,0.93) & (-1.68,-0.34)   \\
\noalign{\smallskip}
\end{tabular}
\end{center}
\end{table}

\newpage 
\bibliographystyle{asa}
\bibliography{msm,iv,owlsurvival,survtrees,survtrees2,causal,causal2}

\begin{thebibliography}{59}
\newcommand{\enquote}[1]{``#1''}
\expandafter\ifx\csname natexlab\endcsname\relax\def\natexlab#1{#1}\fi

\bibitem[{Ai and Chen(2003)}]{ai2003efficient}
Ai, C. and Chen, X. (2003), \enquote{Efficient estimation of models with
  conditional moment restrictions containing unknown functions,}
  \textit{Econometrica}, 71, 1795--1843.

\bibitem[{Andrews(2017)}]{andrews2017examples}
Andrews, D.~W. (2017), \enquote{Examples of L2-complete and boundedly-complete
  distributions,} \textit{Journal of Econometrics}, 199, 213--220.

\bibitem[{Bickel and Ritov(2003)}]{bickel2003}
Bickel, P.~J. and Ritov, Y. (2003), \enquote{{Nonparametric estimators which
  can be "plugged-in"},} \textit{The Annals of Statistics}, 31, 1033 -- 1053.

\bibitem[{Canay et~al.(2013)Canay, Santos, and Shaikh}]{canay2013testability}
Canay, I.~A., Santos, A., and Shaikh, A.~M. (2013), \enquote{On the testability
  of identification in some nonparametric models with endogeneity,}
  \textit{Econometrica}, 81, 2535--2559.

\bibitem[{Carrasco et~al.(2007)Carrasco, Florens, and
  Renault}]{carrasco2007linear}
Carrasco, M., Florens, J.-P., and Renault, E. (2007), \enquote{Linear inverse
  problems in structural econometrics estimation based on spectral
  decomposition and regularization,} \textit{Handbook of Econometrics}, 6,
  5633--5751.

\bibitem[{Chen et~al.(2014)Chen, Chernozhukov, Lee, and Newey}]{chen2014local}
Chen, X., Chernozhukov, V., Lee, S., and Newey, W.~K. (2014), \enquote{Local
  identification of nonparametric and semiparametric models,}
  \textit{Econometrica}, 82, 785--809.

\bibitem[{Connors et~al.(1996)Connors, Speroff, Dawson, Thomas, Harrell,
  Wagner, Desbiens, Goldman, Wu, Califf, Fulkerson, Vidaillet, Broste, Bellamy,
  Lynn, and Knaus}]{5c6af36c0fb64cfcbb482d75c2bc7ff1}
Connors, A., Speroff, T., Dawson, N., Thomas, C., Harrell, F., Wagner, D.,
  Desbiens, N., Goldman, L., Wu, A., Califf, R., Fulkerson, W., Vidaillet, H.,
  Broste, S., Bellamy, P., Lynn, J., and Knaus, W. (1996), \enquote{The
  effectiveness of right heart catheterization in the initial care of
  critically ill patients,} \textit{JAMA - Journal of the American Medical
  Association}, 276, 889--897.

\bibitem[{Cui and Tchetgen~Tchetgen(2019)}]{cui2019selective}
Cui, Y. and Tchetgen~Tchetgen, E. (2019), \enquote{Selective machine learning
  for doubly robust functionals,} \textit{arXiv preprint arXiv:1911.02029}.

\bibitem[{Dagan et~al.(2021)Dagan, Barda, Kepten, Miron, Perchik, Katz,
  Hern{\'a}n, Lipsitch, Reis, and Balicer}]{dagan2021bnt162b2}
Dagan, N., Barda, N., Kepten, E., Miron, O., Perchik, S., Katz, M.~A.,
  Hern{\'a}n, M.~A., Lipsitch, M., Reis, B., and Balicer, R.~D. (2021),
  \enquote{BNT162b2 mRNA Covid-19 vaccine in a nationwide mass vaccination
  setting,} \textit{New England Journal of Medicine}.

\bibitem[{Darolles et~al.(2011)Darolles, Fan, Florens, and
  Renault}]{darolles2011nonparametric}
Darolles, S., Fan, Y., Florens, J.-P., and Renault, E. (2011),
  \enquote{Nonparametric instrumental regression,} \textit{Econometrica}, 79,
  1541--1565.

\bibitem[{D'Haultfoeuille(2011)}]{d2011completeness}
D'Haultfoeuille, X. (2011), \enquote{On the completeness condition in
  nonparametric instrumental problems,} \textit{Econometric Theory}, 460--471.

\bibitem[{Flanders et~al.(2011)Flanders, Klein, Darrow, Strickland, Sarnat,
  Sarnat, Waller, Winquist, and Tolbert}]{flanders2011method}
Flanders, W.~D., Klein, M., Darrow, L.~A., Strickland, M.~J., Sarnat, S.~E.,
  Sarnat, J.~A., Waller, L.~A., Winquist, A., and Tolbert, P.~E. (2011),
  \enquote{A method for detection of residual confounding in time-series and
  other observational studies,} \textit{Epidemiology (Cambridge, Mass.)}, 22,
  59.

\bibitem[{Flanders et~al.(2015)Flanders, Strickland, and
  Klein}]{Flanders2015ANM}
Flanders, W.~D., Strickland, M.~J., and Klein, M. (2015), \enquote{A New Method
  for Partial Correction of Residual Confounding in Time-Series and Other
  Observational Studies.} \textit{American Journal of Epidemiology}, 185 10,
  941--949.

\bibitem[{Gagnon-Bartsch and Speed(2012)}]{GagnonBartsch2012UsingCG}
Gagnon-Bartsch, J.~A. and Speed, T.~P. (2012), \enquote{Using control genes to
  correct for unwanted variation in microarray data.} \textit{Biostatistics},
  13 3, 539--52.

\bibitem[{Ghassami et~al.(2022)Ghassami, Ying, Shpitser, and
  Tchetgen}]{ghassami2022minimax}
Ghassami, A., Ying, A., Shpitser, I., and Tchetgen, E.~T. (2022),
  \enquote{Minimax Kernel Machine Learning for a Class of Doubly Robust
  Functionals with Application to Proximal Causal Inference,} in
  \textit{International Conference on Artificial Intelligence and Statistics},
  PMLR, pp. 7210--7239.

\bibitem[{Hahn(1998)}]{hahn1998role}
Hahn, J. (1998), \enquote{On the role of the propensity score in efficient
  semiparametric estimation of average treatment effects,}
  \textit{Econometrica}, 315--331.

\bibitem[{Hall and Horowitz(2005)}]{hall2005nonparametric}
Hall, P. and Horowitz, J.~L. (2005), \enquote{Nonparametric methods for
  inference in the presence of instrumental variables,} \textit{The Annals of
  Statistics}, 33, 2904--2929.

\bibitem[{Heckman et~al.(1998)Heckman, Ichimura, Smith, Todd,
  et~al.}]{heckman1998characterizing}
Heckman, J., Ichimura, H., Smith, J., Todd, P., et~al. (1998),
  \enquote{Characterizing Selection Bias Using Experimental Data,}
  \textit{Econometrica}, 66, 1017--1098.

\bibitem[{Hern{\'a}n and Robins(2020)}]{hernan2020causal}
Hern{\'a}n, M.~A. and Robins, J.~M. (2020), \enquote{Causal inference: what
  if,} \textit{Boca Raton: Chapman \& Hill/CRC}, 2020.

\bibitem[{Hirano and Imbens(2001)}]{Hirano2001}
Hirano, K. and Imbens, G.~W. (2001), \enquote{Estimation of Causal Effects
  using Propensity Score Weighting: An Application to Data on Right Heart
  Catheterization,} \textit{Health Services and Outcomes Research Methodology},
  2, 259--278.

\bibitem[{Horowitz(2011)}]{horowitz2011applied}
Horowitz, J.~L. (2011), \enquote{Applied nonparametric instrumental variables
  estimation,} \textit{Econometrica}, 79, 347--394.

\bibitem[{Kallus et~al.(2021)Kallus, Mao, and Uehara}]{kallus2021causal}
Kallus, N., Mao, X., and Uehara, M. (2021), \enquote{Causal inference under
  unmeasured confounding with negative controls: A minimax learning approach,}
  \textit{arXiv preprint arXiv:2103.14029}.

\bibitem[{Kang et~al.(2007)Kang, Schafer, et~al.}]{kang2007demystifying}
Kang, J.~D., Schafer, J.~L., et~al. (2007), \enquote{Demystifying double
  robustness: A comparison of alternative strategies for estimating a
  population mean from incomplete data,} \textit{Statistical Science}, 22,
  523--539.

\bibitem[{Kim(2015)}]{kim2015ppcor}
Kim, S. (2015), \enquote{ppcor: an R package for a fast calculation to
  semi-partial correlation coefficients,} \textit{Communications for
  statistical applications and methods}, 22, 665.

\bibitem[{Kompa et~al.(2022)Kompa, Bellamy, Kolokotrones, Robins, and
  Beam}]{kompa2022deep}
Kompa, B., Bellamy, D.~R., Kolokotrones, T., Robins, J.~M., and Beam, A.~L.
  (2022), \enquote{Deep learning methods for proximal inference via maximum
  moment restriction,} \textit{arXiv preprint arXiv:2205.09824}.

\bibitem[{Kress(1989)}]{kress1989linear}
Kress, R. (1989), \textit{Linear integral equations}, vol.~82, Springer.

\bibitem[{Kuroki and Pearl(2014)}]{kuroki2014measurement}
Kuroki, M. and Pearl, J. (2014), \enquote{Measurement bias and effect
  restoration in causal inference,} \textit{Biometrika}, 101, 423--437.

\bibitem[{Li et~al.(2022)Li, Shi, Miao, and Tchetgen}]{li2022double}
Li, K.~Q., Shi, X., Miao, W., and Tchetgen, E.~T. (2022), \enquote{Double
  Negative Control Inference in Test-Negative Design Studies of Vaccine
  Effectiveness,} \textit{ArXiv}.

\bibitem[{Lipsitch et~al.(2010)Lipsitch, Tchetgen, and
  Cohen}]{lipsitch2010negative}
Lipsitch, M., Tchetgen, E.~T., and Cohen, T. (2010), \enquote{Negative
  controls: a tool for detecting confounding and bias in observational
  studies,} \textit{Epidemiology (Cambridge, Mass.)}, 21, 383.

\bibitem[{Mastouri et~al.(2021)Mastouri, Zhu, Gultchin, Korba, Silva, Kusner,
  Gretton, and Muandet}]{mastouri2021proximal}
Mastouri, A., Zhu, Y., Gultchin, L., Korba, A., Silva, R., Kusner, M., Gretton,
  A., and Muandet, K. (2021), \enquote{Proximal causal learning with kernels:
  Two-stage estimation and moment restriction,} in \textit{International
  Conference on Machine Learning}, PMLR, pp. 7512--7523.

\bibitem[{Miao et~al.(2018)Miao, Geng, and
  Tchetgen~Tchetgen}]{miao2018identifying}
Miao, W., Geng, Z., and Tchetgen~Tchetgen, E.~J. (2018), \enquote{Identifying
  causal effects with proxy variables of an unmeasured confounder,}
  \textit{Biometrika}, 105, 987--993.

\bibitem[{Miao et~al.(2022)Miao, Hu, Ogburn, and Zhou}]{miao2022}
Miao, W., Hu, W., Ogburn, E.~L., and Zhou, X.-H. (2022), \enquote{Identifying
  Effects of Multiple Treatments in the Presence of Unmeasured Confounding,}
  \textit{Journal of the American Statistical Association}, 0, 1--15.

\bibitem[{Miao and Tchetgen~Tchetgen(2018)}]{miao2018confounding}
Miao, W. and Tchetgen~Tchetgen, E. (2018), \enquote{A Confounding Bridge
  Approach for Double Negative Control Inference on Causal Effects (Supplement
  and Sample Codes are included),} \textit{arXiv preprint arXiv:1808.04945}.

\bibitem[{Newey and McFadden(1994)}]{NEWEY19942111}
Newey, W.~K. and McFadden, D. (1994), \enquote{Chapter 36 Large sample
  estimation and hypothesis testing,} Elsevier, vol.~4 of \textit{Handbook of
  Econometrics}, pp. 2111 -- 2245.

\bibitem[{Newey and Powell(2003)}]{newey2003instrumental}
Newey, W.~K. and Powell, J.~L. (2003), \enquote{Instrumental variable
  estimation of nonparametric models,} \textit{Econometrica}, 71, 1565--1578.

\bibitem[{Olson et~al.(2022)Olson, Newhams, Halasa, Price, Boom, Sahni,
  Pannaraj, Irby, Walker, Schwartz, et~al.}]{olson2022effectiveness}
Olson, S.~M., Newhams, M.~M., Halasa, N.~B., Price, A.~M., Boom, J.~A., Sahni,
  L.~C., Pannaraj, P.~S., Irby, K., Walker, T.~C., Schwartz, S.~P., et~al.
  (2022), \enquote{Effectiveness of BNT162b2 vaccine against critical Covid-19
  in adolescents,} \textit{New England Journal of Medicine}.

\bibitem[{Patel et~al.(2020)Patel, Jackson, and
  Ferdinands}]{patel2020postlicensure}
Patel, M.~M., Jackson, M.~L., and Ferdinands, J. (2020), \enquote{Postlicensure
  evaluation of COVID-19 vaccines,} \textit{JAMA}, 324, 1939--1940.

\bibitem[{Robins(1986)}]{robins1986new}
Robins, J. (1986), \enquote{A new approach to causal inference in mortality
  studies with a sustained exposure period—application to control of the
  healthy worker survivor effect,} \textit{Mathematical modelling}, 7,
  1393--1512.

\bibitem[{Robins et~al.(2017)Robins, Li, Mukherjee, Tchetgen, and van~der
  Vaart}]{robins2017minimax}
Robins, J.~M., Li, L., Mukherjee, R., Tchetgen, E.~T., and van~der Vaart, A.
  (2017), \enquote{{Minimax estimation of a functional on a structured
  high-dimensional model},} \textit{The Annals of Statistics}, 45, 1951 --
  1987.

\bibitem[{Robins et~al.(2008)Robins, Li, Tchetgen, and van~der
  Vaart}]{robins2008HOIF}
Robins, J.~M., Li, L., Tchetgen, E.~T., and van~der Vaart, A. (2008),
  \enquote{Higher order influence functions and minimax estimation of nonlinear
  functionals,} \textit{IMS Collections: Probability and Statistics: Essays in
  Honor of David A. Freedman}, 2, 335--421.

\bibitem[{Robins et~al.(1994)Robins, Rotnitzky, and Zhao}]{robinsetal1994}
Robins, J.~M., Rotnitzky, A., and Zhao, L.~P. (1994), \enquote{Estimation of
  Regression Coefficients When Some Regressors are not Always Observed,}
  \textit{Journal of the American Statistical Association}, 89, 846--866.

\bibitem[{Rotnitzky and Robins(1997)}]{rotnitzky1997}
Rotnitzky, A. and Robins, J. (1997), \enquote{Analysis of semi-parametric
  regression models with non-ignorable non-response,} \textit{Statistics in
  Medicine}, 16, 81--102.

\bibitem[{Scharfstein et~al.(1999)Scharfstein, Rotnitzky, and
  Robins}]{Scharfstein1999}
Scharfstein, D.~O., Rotnitzky, A., and Robins, J.~M. (1999), \enquote{Adjusting
  for Nonignorable Drop-Out Using Semiparametric Nonresponse Models,}
  \textit{Journal of the American Statistical Association}, 94, 1096--1120.

\bibitem[{Shi et~al.(2019)Shi, Miao, Nelson, and
  Tchetgen~Tchetgen}]{Shi2019MultiplyRC}
Shi, X., Miao, W., Nelson, J.~C., and Tchetgen~Tchetgen, E.~J. (2019),
  \enquote{Multiply Robust Causal Inference with Double Negative Control
  Adjustment for Categorical Unmeasured Confounding,} \textit{Journal of the
  Royal Statistical Society: Series B (Statistical Methodology)}, to appear.

\bibitem[{Shi et~al.(2020)Shi, Miao, and Tchetgen~Tchetgen}]{shi2020selective}
Shi, X., Miao, W., and Tchetgen~Tchetgen, E. (2020), \enquote{A Selective
  Review of Negative Control Methods in Epidemiology.} \textit{Current
  Epidemiology Reports}.

\bibitem[{Singh(2020)}]{singh2020kernel}
Singh, R. (2020), \enquote{Kernel methods for unobserved confounding: Negative
  controls, proxies, and instruments,} \textit{arXiv preprint
  arXiv:2012.10315}.

\bibitem[{Sofer et~al.(2016)Sofer, Richardson, Colicino, Schwartz, and
  Tchetgen~Tchetgen}]{sofer2016negative}
Sofer, T., Richardson, D.~B., Colicino, E., Schwartz, J., and
  Tchetgen~Tchetgen, E.~J. (2016), \enquote{On negative outcome control of
  unobserved confounding as a generalization of difference-in-differences,}
  \textit{Statistical Science}, 31, 348.

\bibitem[{Stefanski and Boos(2002)}]{Stefanski2002M}
Stefanski, L.~A. and Boos, D.~D. (2002), \enquote{The Calculus of
  M-Estimation,} \textit{The American Statistician}, 56, 29--38.

\bibitem[{Stephens et~al.(2014)Stephens, Tchetgen, and
  De~Gruttola}]{stephens2014locally}
Stephens, A., Tchetgen, E.~T., and De~Gruttola, V. (2014), \enquote{Locally
  efficient estimation of marginal treatment effects when outcomes are
  correlated: is the prize worth the chase?} \textit{The International Journal
  of Biostatistics}, 10, 59--75.

\bibitem[{Tan(2006)}]{tan2006distributional}
Tan, Z. (2006), \enquote{A Distributional Approach for Causal Inference Using
  Propensity Scores,} \textit{Journal of the American Statistical Association},
  101, 1619--1637.

\bibitem[{Tan(2019{\natexlab{a}})}]{tan2019model}
--- (2019{\natexlab{a}}), \enquote{Model-assisted inference for treatment
  effects using regularized calibrated estimation with high-dimensional data,}
  \textit{Annals of Statistics}, to appear.

\bibitem[{Tan(2019{\natexlab{b}})}]{tan2019regularized}
--- (2019{\natexlab{b}}), \enquote{Regularized calibrated estimation of
  propensity scores with model misspecification and high-dimensional data,}
  \textit{Biometrika}, to appear.

\bibitem[{Tchetgen~Tchetgen(2014)}]{tchetgen2013}
Tchetgen~Tchetgen, E. (2014), \enquote{The Control Outcome Calibration Approach
  for Causal Inference With Unobserved Confounding,} \textit{American Journal
  of Epidemiology}, 179, 633--640.

\bibitem[{Tchetgen~Tchetgen et~al.(2020)Tchetgen~Tchetgen, Ying, Cui, Shi, and
  Miao}]{tchetgen2020}
Tchetgen~Tchetgen, E., Ying, A., Cui, Y., Shi, X., and Miao, W. (2020),
  \enquote{An Introduction to Proximal Causal Learning,} \textit{arXiv preprint
  arXiv:2009.10982}.

\bibitem[{Thompson et~al.(2021)Thompson, Stenehjem, Grannis, Ball, Naleway,
  Ong, DeSilva, Natarajan, Bozio, Lewis, et~al.}]{thompson2021effectiveness}
Thompson, M.~G., Stenehjem, E., Grannis, S., Ball, S.~W., Naleway, A.~L., Ong,
  T.~C., DeSilva, M.~B., Natarajan, K., Bozio, C.~H., Lewis, N., et~al. (2021),
  \enquote{Effectiveness of Covid-19 vaccines in ambulatory and inpatient care
  settings,} \textit{New England Journal of Medicine}, 385, 1355--1371.

\bibitem[{Van~der Vaart(1998)}]{vaart_1998}
Van~der Vaart, A.~W. (1998), \textit{Asymptotic Statistics}, Cambridge Series
  in Statistical and Probabilistic Mathematics, Cambridge University Press.

\bibitem[{Vermeulen and Vansteelandt(2015)}]{2015biasreduce}
Vermeulen, K. and Vansteelandt, S. (2015), \enquote{Bias-Reduced Doubly Robust
  Estimation,} \textit{Journal of the American Statistical Association}, 110,
  1024--1036.

\bibitem[{Wang et~al.(2017)Wang, Zhao, Hastie, and Owen}]{wang2017}
Wang, J., Zhao, Q., Hastie, T., and Owen, A.~B. (2017), \enquote{Confounder
  adjustment in multiple hypothesis testing,} \textit{Ann. Statist.}, 45,
  1863--1894.

\bibitem[{White(1982)}]{white1982maximum}
White, H. (1982), \enquote{Maximum likelihood estimation of misspecified
  models,} \textit{Econometrica: Journal of the Econometric Society}, 1--25.

\end{thebibliography}

\end{document}